\documentclass[12pt]{article}
\usepackage{eurosym}
\usepackage{amsfonts}
\usepackage{amssymb}
\usepackage{graphicx}
\usepackage{amsmath}
\usepackage{amstext}
\usepackage{natbib}
\usepackage{geometry}
\usepackage{tikz}
\usepackage[onehalfspacing]{setspace}
\usepackage{booktabs}
\usepackage[ruled]{algorithm2e}
\usepackage[labelfont=sc,font=small]{caption}
\usepackage{subcaption}
\usepackage{amsthm}
\usepackage{hyperref}
\usepackage{bbm}

\usetikzlibrary{positioning}
\usetikzlibrary{shapes}
\usetikzlibrary{calc}
\newtheoremstyle{break}
  {}
  {}
  {\itshape}
  {}
  {\bfseries}
  {}
  {\newline}
  {}
\theoremstyle{break}
\SetAlFnt{\small}
\SetAlCapFnt{\small}
\SetAlCapNameFnt{\small}
\SetAlCapHSkip{0pt}
\IncMargin{-\parindent}
\newcommand{\off}{\textnormal{off}}
\newcommand{\on}{\textnormal{on}}
\newcommand{\uv}{\underline{v}}
\newcommand{\bv}{\bar{v}}

\newtheorem{theorem}{Theorem}

\newtheorem{definition}{Definition}

\newtheorem{proposition}{Proposition}

\begin{document}

\title{How Do Digital Advertising Auctions \\
Impact Product Prices?\thanks{%
We acknowledge financial support through NSF Grant SES-1948692, the Omidyar
Network, and the Sloan Foundation. An earlier version of this paper
circulated under the title ``Managed Campaigns and Data-Augmented Auctions
for Digital Advertising.'' An extended abstract appears in the
Proceedings of the 24th ACM Conference on Economics and Computation (EC
'23). We thank the editor and four referees for productive suggestions. We
thank Santiago Balseiro, Gianluca Binelli, and Balu Sivan for helpful
comments, and Roi Orzach for excellent research assistance.}}
\author{Dirk Bergemann\thanks{%
Department of Economics, Yale University, New Haven, CT 06511,
dirk.bergemann@yale.edu} \and Alessandro Bonatti\thanks{%
Sloan School of Management, MIT, Cambridge, MA 02142, bonatti@mit.edu} \and %
Nicholas Wu\thanks{%
Department of Economics, Yale University, New Haven, CT 06511,
nick.wu@yale.edu}}
\date{\today }
\maketitle

\begin{abstract}

We present a model of digital advertising with three key features: (i) advertisers can reach consumers on and off a platform, (ii) additional data enhances the value of advertiser-consumer matches, and (iii) bidding follows  auction-like mechanisms. We contrast data-augmented auctions, which leverage the platform’s data advantage to improve match quality, and managed campaign mechanisms that automate match formation and price-setting.

The platform-optimal mechanism is a managed campaign that conditions on-platform prices for sponsored products on the off-platform prices set by all advertisers. This mechanism yields  the efficient on-platform allocation but inefficient off-platform allocations due to high product prices; it attains the vertical integration profit for the platform and advertisers; and it increases off-platform product prices and decreases consumer surplus, relative to data-augmented auctions.\bigskip

\bigskip

\bigskip\bigskip
\end{abstract}



\newpage

\section{Introduction}


\subsection{Motivation}

Digital advertising facilitates the matching of consumers and advertisers
online. Large platforms utilize their extensive consumer data to connect
online shoppers with their preferred firms and products. In turn,
advertisers join these platforms in order to target a wide range of
potential consumers beyond their existing customer base. As a result, sponsored content is ubiquitous on the Internet: advertising makes up nearly all the revenue of search engines and social media platforms, a growing fraction of the revenue of retail platforms such as Amazon and Instacart, and a large fraction of other retail platforms' revenue, such as Alibaba's Taobao marketplace.

The role of these platforms'  proprietary datasets becomes apparent when we decompose the value of digital paid traffic across the web. Google, Meta, Amazon, and other platforms place advertising \emph{directly} on
their own sites (through sponsored search results, stories, and products) and also serve as \emph{intermediaries} that  place advertising on third-party sites. The most recent data from public filings show that
nearly $60\%$ of the worldwide digital advertising revenue (which exceeds $\$600$ billion) accrued on the  platforms' own websites \citep{lebo23}. In particular, Google received revenues of $\$191$ billion from digital advertising on its own sites, (e.g., google.com, youtube.com, etc.)  and merely $\$33$ billion from ad placement on third-party websites.\footnote{See Google's 2023 financial report to the Securities and Exchange Commission, available at\\ \url{https://www.statista.com/study/163755/alphabet-google-annual-report-2023}.} Thus, the majority of revenue accrues precisely where
the deployment of proprietary data is completely unrestricted and not
accessible to competing marketers. Indeed, the evolution of this marketplace
suggests a reversal in the traditional assumptions on asymmetric information
in digital advertising.

As the market for digital advertising has grown and become more complex, the prevailing mechanisms by which platforms sell ads have also shifted. Digital platforms increasingly act as
intermediaries that run managed campaigns for advertisers, who  set a fixed budget, specify high-level objectives for their campaigns, and leave the task of bidding to ``auto-bidding'' algorithms offered by the platform. 
 The most recent estimates suggest that
over $80\%$ of digital advertising is now generated by managed campaigns \citep{dgjl22,dmmz22}. For example, over $80\%$
of Google  advertisers were using automated bidding in 2023.\footnote{%
See \url{https://ads.google.com/home/measurement/bidding/}. Furthermore, \citet[\S 5.201, \S 5.76]{cma20} reported that 40-50\% of Google's 2019 search
advertising revenue in the UK came from advertisers using automated bidding. Finally, 90-100\% of UK advertisers on Facebook were using the default auto-bidding feature, which  does not allow advertisers to specify a maximum bid.}

In this paper, we provide an equilibrium treatment of how data-intensive mechanisms for selling advertising impact product prices and welfare both on and off the platform.  Our approach takes two fundamental aspects of digital
advertising into account. First, consistent with the revenue breakdown described above, platforms
possess valuable data that can enhance the matching efficiency. Second, advertisers have parallel sales channels, i.e., they can reach their customers on and off the platform. 

We consider a monopolist digital
platform that sells access to its users. Advertisers determine their pricing strategy on and off the platform and their advertising strategy on the
platform. On-platform consumers act as \emph{shoppers} and choose the product
that offers the highest net value. Because these consumers compare the advertised offers to all firms' off-platform prices, advertisers endogenously behave as if under a ``showrooming''
constraint: they wish to ensure their on-platform offers are at least as attractive as their off-platform offers. Conversely, consumers off the platform are \emph{loyal} and buy from a single brand. Consequently, advertisers face a trade-off between setting optimal prices for their loyal
customers off-platform and the option of charging higher personalized prices to
on-platform shoppers.\footnote{%
In pure advertising platforms, where the matching fee is typically incurred
before the transaction (e.g., through pay-per-impression or pay-per-click
fees), the advertiser faces the showrooming constraint directly. The advertiser wants to pay for the listing only if it leads to a sale, as the offline transaction would have occurred without the advertising. In platforms where the fee is based on transactions, such as referral fees on shopping services like Amazon, the platform often imposes the showrooming
constraint through a most favored nation clause. This clause requires the
advertiser to offer the most favorable price online.}

A key innovation in our model is that the platform actively influences the firms' advertising campaigns. With access to the platform's data, advertisers can offer prices that reflect the consumers' willingness to pay. This form of price discrimination broadens the market and enhances the
efficiency of matching on the platform. 
Off the platform, advertisers lack additional data
and offer a uniform price. 
We contrast two main mechanisms for allocating advertising space on the platform: \emph{data-augmented auctions} and \emph{managed campaigns}. Both these mechanisms, as well as our simple advertising model, are of course simplifications. Throughout the paper, we discuss the relevance of each of our modeling choices to real-world advertising markets: our model of advertised prices in Section \ref{sec:model}; the auction format in Section \ref{section:dab}; and the managed-campaign mechanism in Section \ref{sec:managed_campaign}.

Our model demonstrates how  any
analysis of the pass-through of online advertising costs must account for
cross-channel distortions. Indeed, we show that advertisers raise prices off
the platform to gain a competitive edge on the platform. In particular, under the platform-optimal mechanism, the higher costs of online advertising are passed on to consumers by means of higher product prices \emph{off} rather than \emph{on} the platform.

\subsection{Results}

We begin our analysis with a second-price auction for a single advertising slot where the platform augments the bidders' information by soliciting bids based on the match values with each  consumer. We refer
to this as \emph{data-augmented bidding}: each advertiser submits a bid for the slot and a price at which to offer its product if it wins the slot.

We derive the optimal bidding and pricing strategy of the advertisers. On the platform, the second-price auction
implements an efficient allocation, and the additional data allows the advertisers to sell successfully to consumers with lower values without the need to price them out of the market (Theorem \ref{thm:bidding_eq}). Additionally, each advertiser must set the price at which to offer its product to \emph{loyal customers} off the platform. In equilibrium, the advertisers raise their off-platform prices, relative to the prices they would have charged in a stand-alone market (Proposition \ref{prop:bidding_price_larger}): by
offering their product only at a higher price, each advertiser can weaken the showrooming constraint and extract more surplus on the platform. Consequently, the off-platform prices increase with the number of on-platform shoppers.

Next, we introduce the concept of a \emph{managed campaign}. In this more
centralized mechanism, the platform proposes to each advertiser a steering policy and a  pricing policy  for their
product on the platform. Contextually, the platform requests a fixed fee from each advertiser, which we can interpret as a required advertising budget. Each advertiser simultaneously decides whether to enter into the managed campaign or not, and how to price its product off the platform. We show that the platform optimizes its revenue by matching firms and consumers efficiently and by offering a \emph{best-value pricing} policy. This policy ensures the efficient firm always makes the offer with the best value to the consumer, even if its competitors deviate in their posted prices (Theorem \ref{thm:managed_eq_best_value}). In doing so, the
platform weakens competition and leads the firms to raise their posted prices off the platform in order to extract more surplus from online consumers. 

Best-value pricing is not only revenue-optimal for the platform; the joint producer surplus attains the vertical integration benchmark where one
firm controls all the advertisers and the platform (Theorem \ref%
{thm:best_value_max}). In consequence, the posted prices off-platform are higher than under the data-augmented auction (Theorem \ref{thm:managed_prices_highest_variant}). By comparing the  prices charged to consumers and the advertising costs across these two mechanisms, we can then quantify a notion of pass-through (Proposition \ref{prop:passthrough}).

\subsection{Policy Relevance}


The digital platforms that offer managed campaigns also enjoy significant market power, which has raised regulatory concerns. In a recent report, the
UK regulator argues:

\begin{quote}
``Where an advertising platform has market power [...] advertiser bids in
its auctions are higher, resulting in higher prices. In addition, the
platforms may be able to use levers including the use of reserve prices or
mechanisms such as automated bidding to extract more rent from advertisers.
[...]

Higher advertising prices matter because they represent increased costs to
the firms producing goods and services which are purchased by consumers. We
would expect these costs to be passed through to consumers in terms of
higher prices for goods and services, even if the downstream market is
highly competitive." \cite[\S 6.19, \S6.20.]{cma20}
\end{quote}

The  \citet[Chapter 5 and Appendix Q]{cma20} sets forth the principle that platforms should act in customers'
best interests when making choices on their behalf. Our baseline model raises the concern that automated bidding options in Google and Facebook could be used to increase platform revenues to the consumers' detriment instead.

We therefore deploy our model to examine two competition- and privacy-policy interventions. The first policy we consider restricts the platform's auto-bidding algorithms by requiring the pricing and steering policies to be \textit{independent} of all off-platform posted prices. We show that  \emph{any} independent managed campaign that steers consumers efficiently leads to lower on-platform prices than the fully optimal managed campaign  (Theorem \ref{thm:indep_managed_campaign}). In particular, limiting the signals that the pricing policy can use restores  the possibility 
for on-platform consumers to be poached by other firms through off-platform price cuts. This force fosters competition and benefits consumers on both sales channels.


The second policy we consider is a privacy restriction that prevents the
platform from steering consumers and setting prices on the basis of
the consumers' detailed data. Instead, we allow the platform to condition its
steering and pricing decision on the basis of coarse information only, i.e.,
on the identity of each consumer's favorite firm. This restriction is
equivalent to removing the ability to perfectly price discriminate using the
platform's data. In this scenario, the firms sell to both on- and
off-platform consumers via the same posted price (Proposition \ref{prop:partial_privacy}). 
The privacy restriction reduces off-platform prices compared to the benchmark of a  managed campaign, but  may reduce total surplus on the platform, because low-value consumers no longer receive personalized discounts.






\subsection{Related Literature}


Our paper contributes to the literature on online advertising auctions. Recent work in this field studies learning in repeated auctions  \citep{bg2019,kn2020,nccep2022}, discriminatory effects \citep{cmv2019,asbkmr2019,nt2020}, and collusion \citep{dgp2020, drrs2022}. Our focus, instead, is the comparison of auctions with other
allocation mechanisms in a setting with a given information structure and parallel sales channels. As such, our approach is related to, yet distinct from, \cite{bash22}, who compare auctions and auto-bidding mechanisms in a single market under exogenous limits to the ability to steer and to price discriminate. \cite{mope22} study a model of
targeted bidding (i.e., data-augmented auctions) where the  number of organic search results is fixed. In their setting, sponsored content may crowd out organic
information when the same firm wins both types of links. This limits competition, facilitates market segmentation, and reduces welfare.

Several papers \citep{glp2021, lmp2022, meht22, dmmz22} study online auction design in the presence of autobidders and return-on-investment constraints.\footnote{%
A recent literature on auto-bidding algorithms allows for objective functions by the bidders outside of the class of quasilinear utility models common in mechanism design. For example, the bidder may seek to maximize return on investments and have budget or spending constraints. \citet{agga19}, \citet{bals21}, and \citet{deng21} offer excellent introductions this rapidly growing research area.}  Our setting adds a dimension related to advertised prices: firms submit both bids for a sponsored link and tailored prices to offer consumers. While \cite{ll2023} also investigate mechanisms
that allow for advertised prices, we further explore the interaction of these mechanisms with off-platform activity.

Our paper also relates to the literature on information design in auctions and markets. In particular, 
\cite{bebm15}, \cite{hasi22}, and \cite{elgk20} study the effect of market
segmentations and the achievable combinations of consumer and producer
surplus, i.e., how to use data to make markets more or less competitive.

As in \cite{vari80},  the advertisers in our model face two segments of consumers, \emph{shoppers} on the platform and \emph{loyals} off the platform. The design of the auction is therefore subject to the showrooming constraint, i.e., to competition from a separate and distinct market. Earlier papers on  ``partial mechanism
design'' or ``mechanism design with a competitive fringe'' studied  mechanism design in settings where the agents' outside option consists of participating in alternative markets, e.g., %
\citet{phsk12}, \citet{tiro12}, \citet{cade15}, and \citet{fusk15}.
\newpage

The showrooming constraint in our model is related to a growing literature on digital platforms with competing advertisers or multiple sales channels. Recent contributions include \citet{dede16}, \citet{bash20}, 
\citet{mish19}, and \citet{wawr20}. In our setting,  advertisers deter  consumers from showrooming to capture the added value of making data-augmented offers when  selling on the platform. In parallel work, \citet{bebo22} study on- and off-platform competition with multi-product firms and nonlinear pricing. They focus on the
implications of managed campaigns for  equilibrium product quality, relative to our paper's exploration of showrooming and its impact on pricing strategies in the off-platform markets.


Finally, a recent contribution by \cite{varian22} analyzes the relationship between advertising costs and product prices through the lens of a single (representative) online merchant. The size of the advertising audience increases sales proportionally at every price level, with a convex cost of
increasing the audience size. In his separable model, an
exogenous increase in advertising costs does not necessarily lead to an increase in product prices.

\section{Model}

\label{sec:model}

\paragraph{Payoffs and Information}

There are $J$ advertisers (or firms) indexed by $j=1, 2, ..., J$, each selling unique indivisible products and a single digital platform. Each firm's production cost is normalized to zero. There is a unit mass of consumers, each
demanding a single product. The willingness to pay $v_j$ for each firm's
product is drawn independently across consumers and firms according to a
distribution function $F$ that admits a strictly positive density $f$ on its
support $V = [\underline{v},\bar{v}]\subset\mathbb{R}_{+}$. 
The consumer's \textit{value} is given by the vector of willingness to pay 
\begin{equation*}
v = (v_1, ..., v_J) \in V^J\subset\mathbb{R}_{+}^{J}.
\end{equation*}
The utility for a consumer with value $v$ of purchasing product $j$ at price 
$p_j$ is given by 
\begin{equation*}
v_j - p_j.
\end{equation*}
Initially, values are observed by the consumers and by the platform, but not
by the firms.\footnote{
The symmetry in the information is helpful for the welfare comparison but is
clearly a stark assumption. The equilibrium implications are robust to a
more general formulation in which the platform is endowed with partial information.} 
\newpage 

\paragraph{Firms and Platform} The platform presents consumers with a single \textquotedblleft
sponsored\textquotedblright\ result followed by a list of
non-sponsored products. The platform allocates the sponsored position using either a \textit{data-augmented auction} or a \textit{managed campaign}. We
describe these two mechanisms in Sections \ref{section:dab} and 
\ref{sec:managed_campaign}, respectively, and  we connect them to current practices in digital advertising markets. Under either mechanism, an on-platform consumer with value $v$ receives a personalized offer to buy some firm $j$'s product at a price $p_{j}(v)$. Thus, the firm in the sponsored slot can condition its price on the $J$-dimensional consumer value. In addition to the on-platform prices $p_j(v)$, each firm $j$ posts a price $\bar{p}_j$ for its product off the platform.  




\paragraph{On-platform Consumers}

A measure $\lambda \in [0,1]$ of consumers are on-platform ``shoppers.''
These consumers observe $J+1$ prices: the advertised price $p_j(v)$ by the
firm $j$ that is awarded the sponsored slot, as well as the  prices $\bar{p}_k$ posted by all firms $k=1,...,J$. We can view these prices as organic results shown by the platform, or equivalently interpret the model as allowing for free search: only a ``sponsored'' firm can target a price offer to an on-platform consumer, but the consumer can search and find the prices posted by any firm.\footnote{The dual presence of sponsored and organic search describes most closely the practice of platforms such as Google or Amazon. However, in our model, the key feature is having at least some organic search results, not necessarily having both types under one platform. Thus, our model applies  even to social networks like Meta, which have less search functionality compared to Google or Amazon.} 

Under either interpretation, each firm $j$ is subject to a 
\textit{showrooming constraint} when setting its on-platform prices: for all $v$, the prices it advertises on
the platform must satisfy $p_j(v) \le \bar{p}_j$.
Thus, in this model, firms offer the lowest prices on the platform, regardless of whether the platform imposes price-parity or most-favored-nation clauses.\footnote{%
We use the upper bar notation for off-platform prices because the
posted price $\bar{p}_j$ is an upper bound on the amount that any consumer
will pay for firm $j$'s good.}
\paragraph{Off-platform Consumers}

The remaining $1-\lambda$ measure of consumers are ``loyals'' who visit only a single firm off the platform (e.g., its physical store or website).%
\footnote{Section \ref{offplat} extends the model to off-platform competition.} The
off-platform consumer population is divided into $J$ captive segments of
size $(1- \lambda)/J$.  Segment $j$ considers firm $j$ only: these consumers 
buy if and only if the off-platform price $\bar{p}_j$ is lower than their willingness to pay $v_j$.
Figure \ref{fig:model} summarizes our model.

\begin{figure}[htbp]
\centering
\par
\begin{tikzpicture}
\tikzset{elpsb/.style={ellipse,draw=blue!60, fill=blue!5, very thick, minimum width=5em,minimum height=3em,inner ysep=0pt,align=center}}
\tikzset{elpsr/.style={ellipse,draw=red!60, fill=red!5, very thick, minimum width=5em,minimum height=3em,inner ysep=0pt,align=center}}
\tikzset{elpsg/.style={ellipse,draw=green!60, fill=green!5, very thick, minimum width=5em,minimum height=3em,inner ysep=0pt,align=center}}
        
\node (firm)[elpsb] at (0,0)   {Firms};
\node (offer)[elpsg] at (3.5,2) {Custom\\Prices};
\draw[blue,thick,behind path] ($(offer.north west)+(-0.7,0.7)$)  rectangle ($(offer.south east)+(0.7,-0.7)$);
\node (platform) at (3.5,3.5) {Platform};
\node (post)[elpsg] at (3.5,-2) {Posted\\Prices};
\node (onc)[elpsb] at (8,2) {Shoppers};
\node (offc)[elpsb] at (8,-2) {Loyals};

\draw[->,very thick,red] (firm) -- (offer);
\draw[->,very thick,red] (firm) -- (post);
\draw[->,very thick,red] (onc) -- (post);
\draw[->,very thick,red] (onc) -- (offer);
\draw[->,very thick,red] (offc) -- (post);

\node[align=center] at (0.7,-1.4) {own site /\\physical store};

\draw[<-,very thick,dashed,red] (post) -- (offer);
\node at (3.5,0) {showrooming};
\node at (6,-2.5) {(just one)};
\draw[->, thick,black,shorten <=3pt] (onc.north east) -- ++(.5,.5) node[black,right] {$J+1$ prices};
\draw[->, thick,black,shorten <=3pt] (offc.north east) -- ++(.5,.5) node[black,right] {$1$ price};

\end{tikzpicture}
\caption{Model Depiction}
\label{fig:model}
\end{figure}

\paragraph{Digital Advertising through the Lens of the Model}

Digital platforms offer a variety of advertisement formats, such as sponsored links, images, or videos. The content, often a product-price pair selected from the advertiser’s portfolio, can differ across media channels and is tailored to individual consumers. Advertisers face three key decisions: identifying the target users, selecting the appropriate advertisement for each user, and determining the bid for each user’s attention. The best strategy depends on the platform’s nature and the advertiser’s product line. For instance, a brand with multiple product lines would adopt a different approach than a single-product firm, adjusting its campaign according to the type of platform, e.g., search engines, social networks, or third-party publishers.

In our model, each firm offers a single product with fixed characteristics.
Thus, the content of an advertisement is limited to a specific brand and to
a personalized price. As such, our model is certainly an abstraction from the rich practice of digital advertising, where brands have multiple product lines and products have many features. 

At the same time, our model allows us to capture two crucial aspects features of real-world digital advertising: the platform’s ability to match consumers with their preferred firms and the value created through personalized pricing, which offers discounts to lower-value consumers who would not purchase at the monopoly price.

The model also accommodates a broader interpretation where each firm offers a range of products varying in quality and price. The platform’s information enables firms to guide each consumer to a different quality-price pair within their product line, a process known as \emph{product steering}. This process combines value creation and extraction, similar to our single-product model. As the variation in product quality diminishes (i.e., the products of each firm become more alike), product steering becomes akin to personalized pricing.\footnote{See \cite{bash22}, \cite{bebo22}, and \cite{tewr22} for recent models of
product steering on digital platforms.}\newpage

Finally, in our model,  personalized pricing is exclusive to the platform and only applies to the firm that wins the sponsored slot. This is due to the consumers’ unit demand and the significant difference in the information each firm possesses on- and off-platform. However, in real-world scenarios, multiple forms of price discrimination, such as market segmentation and nonlinear pricing, can occur both on and off the platform. In that sense, our model accentuates the differences between these two sales channels.

\section{Data-Augmented Auctions}

\label{section:dab}

In this section, the platform runs an auction to determine which firm makes
a personalized offer to each consumer. The platform provides the advertisers
with information about the characteristics of the consumer, summarized by
the vector of values $v$. Because the advertisers can make their bidding and
pricing decisions contingent on the information disclosed by the platform,
we refer to this mechanism as \emph{data-augmented bidding}.%

\subsection{Data-Augmented Bidding: Mechanism}

The platform runs a second-price auction  for each realized consumer value $v$ separately. In each auction, the platform  enables the advertisers to condition bids and sponsored prices on the consumer's value. Formally, each advertiser $j$ adopts a bidding strategy 
\begin{equation*}
b_{j}:V^{J}\rightarrow \mathbb{R_{+}},
\end{equation*}%
and a sponsored pricing strategy 
\begin{equation*}
p_{j}:V^{J}\rightarrow \mathbb{R_{+}}.
\end{equation*}
This game proceeds as follows. First, all firms (simultaneously) post off-platform prices $\bar{p}_j$. Second,
each firm $j$ submits a bid function $b_j:V^J\to\mathbb{R}_+$ and a
sponsored price function $p_j:V^J\to\mathbb{R}_+$ to the platform. Third, a
consumer value $v$ is realized, and a second-price auction (with no reserve)
determines which firm $j$ and which price $p_j(v)$ are advertised to the
consumer. We characterize the symmetric Bayesian Nash equilibria of the bidding and pricing game among the advertisers. 

\subsection{Data-Augmented Bidding: Practice}

\label{section:dabpractice}

Manual bidding is the original mechanism for selling advertising online and is still in use, though it is becoming less common.\footnote{See, for example, \url{https://support.google.com/google-ads/answer/2390250}. At the
same time, Google also suggests to advertisers that setting bids manually
may result in lower performance, as reported in \url{https://growthmindedmarketing.com/blog/google-ads-mistakes-new-campaigns}.}
When bidding manually, an advertiser typically specifies their willingness to pay for a click on a search result, display ad, or sponsored product listing. The advertisers can also modify their bids and their messages according to the platform's information on each consumer. Thus, the platform monetizes its data through an indirect sale of information \citep{adpf90,bebo19}, whereby
advertisers can act \textit{as if} they had direct access to the consumers' characteristics. In our model, where the platform has complete information about the consumer's preferences, the entire value profile $v$ acts as a \emph{targeting category}.%

The rules for the allocation of sponsored placements vary across digital platforms and publishers. 
Broadly speaking, second-price mechanisms are used by the digital platforms on their own websites, e.g., by Google for sponsored search on google.com \citep{edos07}, by Meta on its social networks, and by Amazon for its sponsored product listings. By contrast, the pricing of display advertising by third party publishers such as nytimes.com or
wsj.com, which is often mediated by the digital platforms, has recently seen a transition from second price auctions to first-price auctions.\footnote{This shift is largely due to the organization of this market through ad exchanges \citep{gwms22}.} 
In what follows, we
focus on the second-price data-augmented auction for its prevalence in digital advertising and its simplicity.
\footnote{%
Yet the details of the auction do not matter for our characterization of
equilibrium bids and prices. Indeed, as the platform enables the firms to
bid in a complete-information auction for each consumer type, both the
winner and the price paid for each $v$ are identical in a first- and
second-price auction. Thus, the platform revenue and the equilibrium posted price are also equivalent for our data-augmented auctions.} 

\subsection{Data-Augmented Bidding: Equilibrium}

To help characterize the firms' equilibrium bidding and pricing strategies
for this setting, we first establish a useful property. Proposition \ref%
{prop:bidding_efficiency} below shows that, regardless of the
prices posted by the firms off the platform, a bidding equilibrium in
undominated strategies results in a \textit{symmetric and efficient}
assignment of on-platform consumers. In other words, each on-platform
consumer sees a sponsored offer from the firm they like best.\newpage


\begin{proposition}[Efficient Bidding Outcome]
\label{prop:bidding_efficiency}

Fix a profile of posted prices $\overline{p}$ and consider an on-platform consumer with value $v$. If $v_j > v_k$, firm $j$ bids at least as much as firm $k$ for consumer $v$ in any bidding equilibrium in undominated strategies.
\end{proposition}

\paragraph{Proof of Proposition \ref{prop:bidding_efficiency}} Fix the vector of posted prices $\bar{p}$ and consider a second-price auction for a consumer with value $v$. In this auction, it is weakly dominant for each firm $j$ to bid up to the price it would charge if it won the sponsored slot, $b_j(v)=p_j(v)$. Having observed all posted prices $\bar{p}$, each firm knows that consumer $v$ has the option to buy from the most attractive
off-platform offer, 
\begin{equation}\label{eq:outside}
\underline{u}(v,\bar{p})\triangleq\max_{k=1,...,J}\left(v_k-\bar{p}_k\right)_{+},
\end{equation}
where $(\cdot)_+$ denotes the nonnegative part throughout the paper.

Given this outside option, firm $j$ can offer consumer $v$ (and therefore bid) $$p_j(v)=b_j(v)=\left(v_j-\underline{u}(v,\bar{p})\right)_+,$$
which in particular implies that the showrooming constraint $p_j(v)\leq\bar{p}_j$ is satisfied.

Because the outside option $\underline{u}$ in \eqref{eq:outside} is common to all firms, the
highest-value firm $j=\arg\max_k v_k$ can offer consumer $v$ the highest price $p_j$ and
still make a sale on the platform. Consequently, the highest-value firm also makes the highest bid $b_j$.\hfill\qed


Proposition \ref{prop:bidding_efficiency}  allows us to
separate the outcome of the bidding stage from the posted prices: the equilibrium matches in the bidding game are invariant with respect to the
posted prices. Our main result in this section (Theorem \ref{thm:bidding_eq}), uses this property to characterize the unique symmetric equilibrium (in undominated strategies) of the data-augmented auctions and the associated posted price $p_B$. 

\begin{theorem}[Symmetric Equilibrium]
\label{thm:bidding_eq} There exists a symmetric equilibrium in undominated strategies. In any such  equilibrium:
\begin{enumerate}
\item Consumer $v$ receives and buys a sponsored offer from $j= \arg \max_k
v_k$.
\item Each firm $k$ posts price $\bar{p}_{k}=p_{B}$ satisfying 
\begin{equation}
(1-\lambda )(1-F(p_{B})-p_{B}f(p_{B}))+\lambda J\int_{p_{B}}^{\bv }F^{J-1}(v-p_{B})dF(v)=0.  \label{eqn:bidding_price}
\end{equation}%
Further, when there is a unique solution to \eqref{eqn:bidding_price}, the
symmetric equilibrium  in undominated strategies is unique.\footnote{%
The proof shows  the symmetric equilibrium profit level is unique and the  symmetric equilibrium posted price is generically unique. When there are multiple equilibrium prices, $p_{B}$ denotes the highest such price.}

\item Firm $j = \arg \max_k v_k$ bids $b_j(v)=p_j(v) = \min(v_j,p_B)$.

\item Firm $j \neq \arg \max_k v_k$ bids $b_j(v)=p_j(v)=(v_j-%
\max_{k}(v_k-p_B)_{+})_+$.
\end{enumerate}
\end{theorem}


The proofs of all results, unless noted otherwise, are collected in the Appendix. To gain intuition for the firm's trade-off when posting a price $\bar{p}_{k}$ , consider the
case of a single bidder $(J=1)$. Because the losing bid is trivially nil,
the profit of the firm is then 
\begin{equation}  \label{eq:mon_pi}
(1-\lambda)(1-F(p))p+\lambda \int_{\underline{v}}^{\bar{v}}\min\{v,p\}dF(v).
\end{equation}
We can therefore view the optimal posted price as solving a monopoly profit
maximization problem, plus a second term capturing the on-platform benefit
which is increasing in $p$.


In our setting with $J$ competing firms, equation \eqref{eqn:bidding_price} is the first-order condition for the competitive analog to the monopoly profit \eqref{eq:mon_pi}. The equilibrium posted price $p_B$ balances the winning firm's profit on the two sales channels. By showrooming, the posted price sets an upper bound on the prices that can be advertised to the on-platform consumers. Therefore, the potential to price discriminate more effectively on-platform pushes firms to raise their posted prices. This effect is captured by the second term in the first-order condition \eqref{eqn:bidding_price}, which is positive.%

While it is intuitive that higher posted prices enable higher advertised
prices, equation \eqref{eqn:bidding_price} illustrates a more
nuanced, important property of data-augmented auctions. A marginal increase in $\bar{p}_j$ above $p_B$ benefits firm $j$ only if (i) a consumer values firm $j$'s product at least $p_B$, and (ii) a consumer values all other brands $k$ less than $v_j-p_B$, so that the second highest bid $b_k(v)$ is nil. If the second condition is not met, i.e., if the auction is sufficiently competitive, the second highest bid is given by $b_k(v)=\bar{p}_j+v_k-v_j$ as in part (4.) of Theorem \ref{thm:bidding_eq}. This bid equals the price that firm $k$ would advertise if it won the auction. 
Critically, firm $k$'s bid increases one-to-one with firm $j$'s  posted price $\bar{p}_j$: firm $k$ bids more aggressively for consumer $v$ when winning the auction would enable it to charge a higher price and still make a sale. Thus,  a higher posted price relaxes showrooming but makes a firm a softer competitor in the more competitive auctions, which dampens the effect of raising the price in the first place.

\subsection{Welfare Implications}

We now discuss the welfare implications of data-augmented bidding. Theorem %
\ref{thm:bidding_eq} shows that the on-platform allocation is socially
efficient: every consumer participates and buys  their favorite
product. Relative to the on-platform channel, the off-platform market
suffers from two sources of inefficiency: first,
consumers are loyal to a random firm, i.e., they might be unaware of the
existence of a firm  they prefer; and second, since firms optimally post
a single off-platform price, those consumers with values below the posted
price do not buy at all.


Turning to the welfare implications for consumers, part (3.) of Theorem \ref{thm:bidding_eq}
shows that  the winning
firm extracts all consumer surplus on-platform, up to the equilibrium posted price. Thus, the expected  surplus of an off-platform and an on-platform consumer are given by
\begin{equation}
    CS_{\off}(p_B) = \int_{p_B}^{\bv} (v-p) \, d F(v),\quad\text{and}\quad CS_{\on}(p_B) = \int_{p_B}^{\bv} (v-p) \, dF^J(v).\label{eq:cscomp}
\end{equation}
On both channels, only consumers with values above $p_B$
obtain a positive surplus. 

To capture the effect of the platform on consumer surplus, we then consider  the posted prices. We first define $p_M$ as the monopoly price for distribution $F$,
\begin{equation}  \label{eqn:offline_price}
p_M \triangleq \arg\max_p p\,(1 - F(p),
\end{equation}
and we assume $p_M>\uv$ throughout.

All firms would post price $p_M$ if  they  had a loyal off-platform population only, as can be seen by setting 
$\lambda=0$ in \eqref{eqn:bidding_price}. For any $\lambda>0$, the
second term on the right-hand side of \eqref{eqn:bidding_price} pushes the
equilibrium price $p_B$ above $p_M$. Formally, part (1.) of Proposition \ref{prop:bidding_price_larger} uses a monotone comparative statics argument (which  we shall invoke repeatedly) to  show that the equilibrium price is increasing in $\lambda$. Therefore, posted prices are larger than the monopoly price $p_M$. We  trace out the welfare implications of this result in part (2.) of  Proposition \ref{prop:bidding_price_larger}.

\begin{proposition}[Posted Prices and Welfare Effects]\label{prop:bidding_price_larger}\strut \vspace{-0.7cm}
\begin{enumerate}
    \item The symmetric equilibrium posted price  $p_B$ is increasing in $\lambda$.
    \item Off-platform per capita total surplus and consumer surplus are both decreasing in $\lambda$. 
\end{enumerate}
\end{proposition}

In traditional models of search stemming from brick and mortar stores with non-posted prices, the increased presence of consumers who obtain more quotes has a positive externality on the other consumers. Our model generates the opposite prediction because  the growth of the platform is unambiguously harmful for off-platform
consumers. 

In contrast, the effect of $\lambda$ on on-platform consumer surplus is more nuanced. Because  every on-platform consumer is matched with their favorite firm (as captured by the distribution $F^J$), the expected welfare of an on-platform consumer in \eqref{eq:cscomp} is always larger than an off-platform consumer's. Moreover, every consumer gains weakly ex post  (after learning their $v$) by joining the platform.  

This creates an important participation externality, however, because more consumers joining the platform increases $\lambda$, which raises all off- and on-platform prices.\footnote{In recent work, \cite{kiph21} and \cite{bebg22} document the externalities that consumers impose on each other through their decisions to share data with a two-sided platform.} As $\lambda\to 1$, the equilibrium posted price $p_B \to \bar{v}$. This means
total surplus is at the first-best level, but the firms extract all consumer surplus on and off the platform.



\section{Managed Advertising Campaigns}

\label{sec:managed_campaign}

We now contrast data-augmented auctions with the more novel 
auto-bidding and managed-campaign mechanisms. In a \emph{managed} advertising campaign, the platform determines which firm wins the sponsored slot for each consumer value and makes an offer to the consumer on behalf of that firm. The platform collects a fixed upfront fee for this service from each participating firm. In turn, the firms relinquish agency over the on-platform allocation process, but
they still collect the resulting revenue and post the off-platform prices.

\subsection{Managed Campaigns: Mechanism}

\label{section:mc}

A managed campaign is a mechanism where the platform
conditions the advertised products and prices on all available information:
the consumer's value $v\in V^J$, the firms' participation decisions in the mechanism $a\in\{0,1\}^J$,
and the posted prices $\bar{p}\in\mathbb{R}^J_+$. We thus consider the following extensive form.

\begin{enumerate}
\item The platform proposes a mechanism $(s, p, T)$ to all firms, where $%
s:V^J\times\{0,1\}^J\times\mathbb{R}_+^J\to J$ is a \textit{steering policy}%
, $p:V^J\times \{0,1\}^J\times\mathbb{R}_+^J\to\mathbb{R}_+$ is a \textit{%
pricing policy}, and $T \in \mathbb{R}^J_+$ is a profile of fixed fees (advertising budgets).

\item The firms simultaneously decide whether to accept $(a_j=1)$ or reject $%
(a_j=0)$ the platform's offer and what off-platform price $\overline{p}_j$
to post.

\item If firm $j$ accepts the platform's offer, it pays a fee $%
T_j$. Its product is offered to a subset of on-platform consumers according
to the steering policy $s$ and priced according to the policy $p$. 
\end{enumerate}

In other words, the steering policy steers each consumer to a firm, depending on all firms' participation decisions, their posted prices, and the consumer's own value. The pricing
policy maps those same variables into an advertised price. In the remainder of this section, we focus on a specific instance of a managed campaign and then show that these pricing and steering policies are revenue optimal for the platform.




\begin{definition}[Best-Value Pricing]
The \textit{best-value pricing} policy sets 
\begin{equation}  \label{eq:best_value_pricing}
p(v, a, \bar{p}) = \min(v_j, \bar{p}_j, \min_{k \neq j}(v_j - v_k + \bar{p}_k))_+,
\end{equation}
where $j=s(v, a, \bar{p})$ is the firm selected by the platform's steering
policy.
\end{definition}

Correspondingly, the efficient steering policy selects the consumer's favorite firm among those that participate in the mechanism.

\begin{definition}[Efficient Steering]
\label{eq:eff_steer} The efficient steering policy sets $s(v,a,\bar{p})=\arg\max_{j}a_{j}v_{j}.$
\end{definition}


When combined with efficient steering, the best-value pricing policy ensures that the sponsored firm's advertised price $p_j(v)$ yields the best value to
each  consumer that likes product $j$ the best, so that no other firm can poach the consumer by posting a lower price $\bar{p}_k$. In this sense, the best-value pricing guarantee \eqref{eq:best_value_pricing} is stronger than a most-favored-nation clause that ensures firms offer their goods at lower prices on- than off-platform: it guarantees that the sponsored firm will compete aggressively with any deviating firm.

\subsection{Managed Campaigns: Practice}

We now discuss the connections between the model and real-world
platforms. Indeed, all three key elements of our model of managed campaigns
connect to the current practices of large digital platforms.

First, relative to data-augmented auctions, ex-ante fixed fees replace individual, per-auction payments. In the
model, the fixed fees represent advertising budgets that firms submit to the
platform. This is the predominant mechanism on pure advertising platforms,
such as Google, Facebook, or Tiktok that match advertisers and consumers,
but do not charge any transaction fees. In all these markets, the firms delegate the spending decisions to the platform, subject to constraints on the returns to their investment.\footnote{Our model is also consistent with this type of arrangement. After Theorem \ref{thm:best_value_max}, we discuss a  sense in which the optimal mechanism promises firms a strictly positive return on their on-platform advertising spending.}

By contrast, retail platform such
as Amazon or Instacart typically receive revenue from a mixture of advertising and sales commissions. Significantly, Amazon's advertising revenue is catching up to those of Google and Meta  \citep{insi23}, which suggests that the relevance of the advertising mechanisms is extending to
retail platforms. Other retail platforms, most notably Alibaba, have very low sales commissions and generate most of their revenue from sponsored listings. In particular, Alibaba's  Taobao operates as a fee-free consumer-to-consumer marketplace where users can pay to rank higher in the search results, thus generating all its revenue from advertising \citep{inve21}.

Second, the platform controls both the allocation of sponsored slots and the prices of the firms' products. Many advertisers run advertising campaigns that target different consumers with promotional offers, which can involve personalized prices, varieties, or product versions. For example, Amazon and Google offer portfolio bidding strategies, which
consist of \textquotedblleft AI-powered, goal-driven bid strategies that
help you optimize bids across multiple campaigns,\textquotedblright\ i.e.,
that choose which offers to target to which user.\footnote{%
See 
\url{https://support.google.com/google-ads/answer/6263058} for Google's description of portfolio bidding and 
\url{https://advertising.amazon.com/blog/introducing-portfolios-for-sponsored-ads} for Amazon's version.} Likewise, Meta's Advantage+ Catalog Ads automatically delivers relevant product recommendations to people based on their revealed  intent. Meta describes this service  as follows.
\begin{quote}
\textquotedblleft You can create a catalog with all your products and create one campaign that drives sales on your website or app. When someone expresses interest in an item from your catalog [or in the types of products or
services you are offering], Meta can dynamically generate an ad for that person and deliver it automatically on mobile, tablet and desktop.\textquotedblright \footnote{%
Meta offers two targeting options that reach customers across different
stages of the buying journey, depending on whether they have manifested
interest in an advertiser's own products or in their industry's products.
See \url{https://www.facebook.com/business/help/397103717129942}.}
\end{quote}

In our model (see the discussion at the end of
Section \ref{sec:model}), each firm sells a single product, and therefore the platform's choice of personalized advertising content reduces to a targeted promotional discount.\newpage

Third, the platform conditions the advertised prices on all off-platform posted prices. 
Real-world managed-campaign algorithms such as Google's Performance Max and Meta's Advantage+ can be viewed as implementing our static mechanism by adapting behavior over time. The connection is as follows. Google's algorithm \textquotedblleft uses Google
AI across bidding, budget optimization, audiences, creatives, attribution, and more.\textquotedblright \footnote{For a detailed description, see
\url{https://support.google.com/google-ads/answer/10724817}.} Thus, the algorithm
adjusts not only the automated bids, but also the creative content shown to
each consumer in order to achieve the campaign goals. For example, if the algorithm detects a drop in clicks on a given advertisement by a certain consumer segment, it can advertise a cheaper product to those consumers, so to improve sales. In practice, this adjustment process occurs gradually.\footnote{%
The initial learning phase of a portfolio advertising campaign mechanism
usually lasts a few weeks. See, for example,
\url{https://support.google.com/google-ads/answer/13020501}. Any subsequent
adjustment period is presumably shorter.} In our static model, the
adjustment is instantaneous: the
platform modifies the advertised prices as soon as a firm deviates from the equilibrium posted price. 

Table \ref{tab:2x2table} offers two  interpretations of our model that summarize the above discussion. The narrow interpretation is the focus of our model, the broader interpretation links our model to a more extensive set of practices and tools. 
\begin{table}[htbp]
\renewcommand{\arraystretch}{1.5} \centering
\begin{tabular}{|c|c|c|}
\hline
& \textbf{Narrow} & \textbf{Broad} \\ \hline
\textbf{Product Line} & Single product & Multiple products \\ \hline
\textbf{Targeting} & Personalized pricing & Product steering \\ \hline
\textbf{Algorithm} & 
\parbox[c]{4.5cm}{\centering Advertised price reacts to
all posted prices} & 
\parbox[c]{6.5cm}{\centering ``Portfolio bids'' and
``catalog ads'' select most profitable product} \\ \hline
\textbf{Timing} & One-time pricing & 
\parbox[c]{5cm}{\centering Learning
over time} \\ \hline
\end{tabular}%
\bigskip
\caption{{\protect\normalsize Mapping the model to real-world managed
campaigns}}
\label{tab:2x2table}
\end{table}



\subsection{Managed Campaigns: Equilibrium with Best-Value Pricing}

A key consequence of best-value pricing and efficient steering is that firms are insulated from competition. We now characterize the symmetric equilibria of the managed campaign with efficient steering and best-value pricing.\newpage

\begin{theorem}[Best-Value Pricing Managed Campaign Equilibrium]
\label{thm:managed_eq_best_value}\strut
In any symmetric equilibrium with efficient steering and
best-value pricing:
\begin{enumerate}
\item On-platform consumers $v$ with $v_j=\max_{k}v_k$ buy from firm $j$ at $%
p_j(v)=\min\{v_j,p_V\}$.

\item The posted price $p_{V}$ is characterized by the following
equation: 
\begin{equation}
(1-\lambda )(1-F(p_{V})-p_{V}f(p_{V}))+\lambda J\int_{p_{V}}^{\bv }F^{J-1}(v)dF(v)=0.\label{eq:managed_price_best_value}\end{equation}%
\end{enumerate}
Further, when there is a unique solution to %
\eqref{eq:managed_price_best_value}, the symmetric equilibrium is unique.\footnote{As in Theorem \ref{thm:bidding_eq}, there is a unique symmetric equilibrium profit level and a generically unique  symmetric equilibrium price. If there are multiple equilibrium prices, $p_V$ denotes the largest such price.}
\end{theorem}

The best-value pricing policy \eqref{eq:best_value_pricing} ensures that each firm  makes a sale to its favorite customers regardless of the posted prices. Therefore, part (1.) shows that each firm sets its posted price like a monopolist with exogenous market segments, subject to  showrooming.

The characterization of the equilibrium $p_V$ in \eqref{eq:managed_price_best_value} follows from the first-order condition for each firm's profit. Indeed, each firm posts a
price that balances the profit off-platform with the relaxation of the showrooming constraint. In particular, the second term in 
\eqref{eq:managed_price_best_value} shows that a marginal  posted price increase yields a one-for-one benefit to firm $j$ when facing any consumer that values product $j$ the most  and values it more than $p_V$.

As it turns out, best-value pricing is revenue-optimal for the platform.
Moreover, it attains an exogenous upper bound on the platform's revenue across all  managed campaigns.

\begin{theorem}[Optimal Managed Campaign]\label{thm:best_value_max}
There exists an equilibrium of the best-value pricing managed campaign with efficient steering that:\vspace{-.05in}
\begin{enumerate}
\item maximizes revenue for the platform among all steering and pricing
policies;

\item attains the integrated (collusive) gross profit for the firms.
\end{enumerate}
\end{theorem}

The argument proceeds by considering the problem of a vertically integrated platform that jointly maximizes the profit of the firms and the platform. The vertically integrated platform can jointly coordinate on-platform and off-platform pricing but still faces the showrooming constraint due to
consumer search. The optimal joint solution is then decentralized by charging a  fixed fee that
extracts all of the firms' surplus, net of an exogenous outside option for the firms.\footnote{Alternatively, we can decompose the advertising budget into a payment per winning bid for each consumer value. In this case, one can show that the bidding algorithm boosts the bids of the advertisers, but never beyond the value of the match. Thus, the auto-bidding mechanism satisfies an ex-post participation constraint for every (winning and losing) bid.}

The firms' outside option consists of not participating in the mechanism and posting a price that maximizes profit on the captive consumers and competes with the advertised prices set by  the best-value pricing policy. This outside option yields a profit level 
\begin{equation}\label{eqn:outside}
\Pi _{O}=\max_{p}\left\{ \frac{1-\lambda }{J}p(1-F(p))+\lambda
\int_{p}^{\bv}pF^{J-1}(v-p)\ dF(v)\right\} .
\end{equation}

The first term in \eqref{eqn:outside} is the profit from
selling to loyal consumers. The second term denotes the
profit the firm makes due to the ability of on-platform consumers to search: upon rejecting the platform's offers, any firm could still make a sale to any on-platform consumer  with a sufficiently high value for its product. In this case, the
best-value pricing policy, which attempts to sell the second-highest valued product to the consumer, charges a  price of zero. The deviating firm then makes the sale when the consumer's value $v$ satisfies $v-p>v^{\prime }$, where $v^{\prime }$ is
the value for the best competitor. The optimal campaign cannot therefore charge firms their entire revenues on the platform. In this sense, the mechanism could be framed  as  delivering a positive return on  investment. 


\subsection{Comparing Advertising Mechanisms}

We now compare the equilibrium posted prices and the welfare implications
under the data-augmented second price auction
and the optimal  managed campaign. We begin with the
comparison of the prices off the platform where we refer to pricing
equations \eqref{eqn:bidding_price} and \eqref{eq:managed_price_best_value}.

\begin{theorem}[Welfare and Posted Price Comparison]
\label{thm:managed_prices_highest_variant} The posted price $p_{V}$ in the
optimal managed campaign is higher than the posted price $p_{B}$ under
data-augmented bidding: 
\begin{equation*}
p_{V}\geq p_{B}\geq p_{M}.
\end{equation*}%
Total consumer surplus and total welfare are lower in the optimal managed
campaign than under data-augmented bidding.
\end{theorem}

In our model,  the impact of digital advertising auctions on product prices is entirely due to the different competitive responses under the data-augmented auctions and the managed campaign mechanism. Theorems \ref{thm:bidding_eq} and \ref{thm:managed_eq_best_value} showed that both data-augmented auctions and the optimal managed
campaign yield an efficient matching of all on-platform consumers to firms.
Moreover, under both mechanisms, the on-platform consumers buy from their
favorite firm (say, $j$) at a price $p_{j}(v)=\min \{v_{j},\bar{p}_{j}\}$.
Therefore, both mechanisms create a common benefit of raising the posted
price $\bar{p}_{j}$, namely to increase revenue on all consumers that like
firm $j$ best \textit{and} value product $j$ more than the posted price.

However, with data-augmented auctions, raising $\bar{p}_j$ raises all
rivals' bids $b_k$ by the same amount, because any firm $k\neq j$ that wins the auction can now charge a higher price and
still induce the consumer to buy its product rather than shop for firm $j$%
's. (Recall the discussion after Theorem \ref{thm:bidding_eq}.) Therefore, raising the posted price $\bar{p}_j$ helps
firm $j$ only if all other bids are nil $(b_{k\neq j}=0)$, which occurs when the consumer $v$ is willing to pay a large enough premium for firm $j$'s product, i.e., $v_j>\bar{p}_j+\max_{k\neq j}v_k$. The latter effect is absent under the optimal managed campaign, where fixed fees replace variable, endogenous payments for each consumer. Therefore, raising posted prices is more profitable under managed campaigns.\footnote{Consistent with this intuition, the proof of Theorem \ref{thm:managed_prices_highest_variant} uses a monotone comparative statics argument that does not require assumptions on the distribution of consumer values.}

Finally, our model admits a formal notion of  ``advertising cost pass-through'' driven by  the mechanisms  for selling ads, namely as platforms move from data-augmented bidding to managed campaigns. To compute this measure of pass-through,  we fix the bargaining power of the platform to that of the sophisticated campaign; that is, the platform charges fees in order to hold firm profit to their outside option. We take these fees as a proxy for advertising cost, and compare the firm net transfer under data-augmented bidding to the transfer under the optimal sophisticated campaign. More precisely, let $T_B$ be the total transfer paid by an individual firm under revenue-maximizing data-augmented bidding,\footnote{To fix the bargaining power of the platform, we consider the data-augmented bidding where the platform can charge a participation fee before bidding occurs. This results in no changes to the pricing outcome and makes the platform's bargaining power comparable.} and let $T_V$ be the total transfer paid by an individual firm under the sophisticated managed campaign. Then we define the ``pass-through'' of the change in mechanisms as:
\begin{equation}
    \eta = \frac{p_V - p_B}{T_V - T_B},
    \label{def:passthrough}
\end{equation}
as how the change in advertising impacts the off-platform posted prices. We can thus characterize the pass-through more formally.

\begin{proposition}[Advertising Mechanism Pass Through]
    \label{prop:passthrough}
    The pass-through rate satisfies $\eta > J$. The increase in advertising costs induced by a managed campaign relative to bidding are reflected by an amplified increase in off-platform prices.
\end{proposition}

Proposition \ref{prop:passthrough} shows that the pass-through of costs is greater the more firms there are. In particular, prices rise more dramatically with more competitors as a result of switching mechanisms precisely because the sophisticated managed campaign softens competition.

\section{Policy Interventions}

In this section, we investigate the impact of potential interventions that a policymaker might impose on the platform. We specifically consider restricting the platform's ability to use sophisticated algorithms that respond to posted prices, and its ability to condition advertising and prices on the consumer' full value profile.

\subsection{Competition Management}

A first regulatory question is whether fully automated systems should be kept in check. The \citet[\S 6.15]{cma20} expresses concerns that ``Although both Google's and
Facebook's core services can be accessed by consumers at no direct cost,
consumers therefore nevertheless suffer financially from the exercise of
market power.'' The alleged concern is that the platform's market power raises the cost of advertising, which is then passed on to consumers.

To address these concerns, we analyze a policy that limits the platform's active role in managing firm competition. In particular, we assume the platform's pricing and steering policies can condition on the consumer's full value profile, but not on the posted prices. 

\paragraph{Current Practice} Limiting auto-bidding algorithms to enable \emph{rule-based bidding} only is an example of such a policy intervention. According to Amazon, rule-based bidding is an existing automated bidding strategy that ``take[s] the guesswork out of adjusting bids,'' but lets advertisers introduce fixed rules for showing creatives to (and bidding on) specific consumer segments.\footnote{See a complete description of available bidding strategies on Amazon at \url{https://advertising.amazon.com/help/GCU2BUWJH2W3A8Z7}.} Likewise, Google's ``Demand Gen'' campaigns allow
advertisers to manually select specific channels for ad display,
\textquotedblleft offering more control over where and how ads appear.\textquotedblright \footnote{See \url{https://ads.google.com/home/campaigns/demand-gen}.} Letting advertisers retain partial control over the bidding rules necessarily slows down the algorithm's adjustment process. In our static model, we capture these algorithms by means of managed campaigns that do not react to evolving market conditions (as proxied by deviations in  posted prices).

\paragraph{Independent Managed Campaign}
We now restrict the platform's pricing and steering policy space by removing the platform's ability to condition on off-platform prices. The platform can now only propose a pricing policy $p:V^{J}\times \{0,1\}^{J} \rightarrow \mathbb{R}_{+}$  and a steering policy $s:V^{J}\times \{0,1\}^{J} \rightarrow J$ that depend  on the value $v$ and the participation decision of the firms.


In the previous section, we showed that the optimal managed campaign dampens competition between firms, resulting in higher posted prices off-platform than in data-augmented bidding. 
Theorem \ref{thm:indep_managed_campaign} shows that forcing  the platform to price 
\textit{independently} of the posted price decisions curtails the ability of the platform to soften
competition. We denote the posted price off the platform induced by the independent campaign by $p_I$.

\begin{theorem}[Independent Managed Campaigns]
\label{thm:indep_managed_campaign} Any independent managed campaign with efficient steering results in lower prices relative to the optimal campaign: $p_I < p_V$, higher social welfare, and higher consumer surplus.
\end{theorem}

The critical economic feature that an independent campaign introduces is the potential for consumers to be poached by other firms. In particular, since the on-platform prices cannot condition on the off-platform prices, a deviating price downwards by firm $i$ can induce some consumers whose favorite firm is $j \neq i$ to buy from $i$ instead. This downward pressure helps mitigate the platform's ability to soften competition.

To illustrate how an independent campaign  restores a downward competitive pressure on prices, we present two examples where the platform steers consumers efficiently and prices as in the optimal managed campaign (i.e. $p_i(v) = \min(v_i, p_V)$). We show that the posted price $p_I$ in the independent campaign can even fall below the monopoly price $p_M$. In such an independent managed campaign, the best-response problem of the firm simplifies to
\[
\max_p \left[ \frac{1-\lambda}{J}p(1-F(p)) + \lambda \int_{\uv}^{\bv} \min(v,p,\hat{p}) F^{J-1}(v - \min(v,p,\hat{p}) + p') dF(v) \right].
\]
The first-order condition for $p$ at or below $\hat{p}$ then implies that: 
\begin{equation}
 \frac{1-\lambda}{J}(1 - F(p) - pf(p)) + \lambda \int_p^{\bv} (F^{J-1}(v) - p dF^{J-1}(v)) dF(v) = 0.
\label{eq:indc}
\end{equation}
Note the presence of the term $p\,dF^{J-1}(v)$, which captures the poaching gain/loss from consumers who at the margin are nearly indifferent between firms. 


\noindent\textbf{Example $J=2: p_I > p_M.$ } Take a uniform distribution of
values ($F(x) = x$), and suppose there is an equal share of on-platform and
off-platform consumers ($\lambda = 1/2$), and consider two firms. Computing the best-response price using the first-order condition (\ref{eq:indc}), we obtain 
\begin{equation*}
p_I = 1 - p_I + 2 \int_{p_I}^1 (v - p_I) \ dv \approx 0.59 > 0.5 = p_M.
\end{equation*}
\textbf{Example: $J=3: p_I < p_M.$ } Consider almost the exact same environment
as the previous example (uniform distribution of values, equal share of
consumers on- and off-platform) but now we consider three firms. In this case, we obtain 
\begin{equation*}
p_I = 1 - p_I + 3 \int_{p_I}^1 (v^2 - p_I (2v)) \ dv \approx 0.43 < 0.5 = p_M.
\end{equation*}

By adding one firm to the previous example, the competitive effect becomes
stronger, and the posted price falls to a level below the monopoly price $p_M$. Note that the pricing in sophisticated managed campaign---which allows for perfect price discrimination up to the showrooming constraint---may not be the optimal independent pricing policy for the platform. In particular, as the example shows, this pricing policy induces stronger competition by firms and sometimes lower posted prices than even the monopoly price. Note that a pricing policy that offers a product for free to consumers whose values are all below $p_M$, and price $p_M$ for a consumer's favorite firm otherwise, can induce an equilibrium  posted price of $p_M$; however, such a pricing policy necessarily concedes rent to the consumer or reduces aggregate welfare. This illustrates the platform's trade-off under the independent pricing restriction: the more aggressive the price discrimination offered by the platform, the stronger the  incentives  for firms to undercut each other, and  the lower the posted prices. However, to raise the posted prices, the platform must concede utility to the consumers; an optimal independent pricing policy must therefore balance these two forces.

\subsection{Privacy and Data}

We now assess the impact of privacy regulation by considering  policies that limit the  firms' access to the consumers'
information. Specifically, we consider cohort-based privacy, which is a restriction in line with the recent Google Privacy Sandbox proposals to replace
third-party cookies. Under this policy, the platform  in our model informs the firms about the consumer's ranking of their products, without disclosing the consumer's exact value for any specific product.\footnote{See the complete Google proposal at \url{https://privacysandbox.com/}. In this Section, we focus on exogenous restrictions on information disclosure. Voluntary information disclosure by the consumer is another important, though different dimension. See \cite{allv19} for a treatment of this  question.}

Formally, the platform's steering policy selects a firm to advertise to each \emph{cohort }of consumers, and each consumer within a cohort has the same preference \emph{ranking} over the $J$ firms. In what follows, we maintain the efficient steering policy (i.e., the platform shows the consumer her favorite product), which yields exactly $J$ distinct consumer cohorts. We then restrict the platform's pricing policy space to $$p:J\times\{0,1\}^{J}\times \mathbb{R}_{+}^{J}\rightarrow \mathbb{R}_{+}.$$
Thus, the platform cannot price based on the consumer's individual value vector $v$, but it can condition the advertised price on the consumer's cohort, the firms' participation  decisions, and the posted prices. This is in contrast to the independent managed campaign, which conditions advertised prices on the
consumer's value but not on posted prices. We denote the resulting equilibrium  off-platform price with \emph{privacy protection} by $p_{P}$.


\begin{proposition}[Cohort Privacy]\label{prop:partial_privacy}\strut
In the platform-optimal managed campaign with cohort privacy, the posted price is $p_{P}$ with: 
\begin{equation}
p_{P}=\frac{1-(1-\lambda )F(p_{P})-\lambda F^{J}(p_{P})}{(1-\lambda
)f(p_{P})+\lambda JF^{J-1}(p_{P})f(p_{P})}.
\label{eqn:partial_privacy_price}
\end{equation}%
This managed campaign can be implemented by the platform pricing each segment at the lowest off-platform price: $p(i,\cdot,\bar{p}) = \min_{i} \bar{p}_i$. On path, the on-platform price is also $p_{P}$, and the equilibrium posted price $p_{P}$ satisfies $p_{M}\le p_{P}\le p_{V}$.
\end{proposition}

Intuitively, firms face a distributional mixture of consumers; a measure $(1-\lambda )/J$ of consumers are loyal with values distributed according to $F$; and a measure of $\lambda /J$ consumers are on-platform shoppers who are matched to the efficient firm, i.e., their values are distributed as $F^{J}$. Hence, the firm would like to be able to
set higher prices to take advantage of a more favorable distribution of
consumer values, but showrooming limits its ability to do so. 

Proposition \ref{prop:partial_privacy} also shows that the off-platform posted
prices are lower than under the optimal managed campaign, which implies greater consumer surplus and total welfare off the platform. However, the  inability to
price discriminate on-platform means the privacy restriction reduces total welfare on the platform too, because low-value consumers are priced out. Hence, as the platform size $\lambda$ grows, total welfare can be worse under privacy restrictions than under  data-augmented bidding or managed campaigns. However, consumer surplus grows because low-value consumers' surplus is nil  in
both settings without privacy protection; that is, the loss in welfare comes entirely from reduced producer surplus.

\section{Extensions and Robustness}

We discuss two extensions of the basic model that speak to the robustness of
the analysis. The first variation concerns the nature of the off-platform
market, the second the nature of the platform, in particular the revenue
model of the platform.

\subsection{Off-Platform Competition}\label{offplat}

Previously, we modeled each firm as operating as a monopolist of a market segment off the
platform. We now show that the analysis and consequent results extend to a
more competitive structure in the off-platform markets. Suppose then that
the off-platform market is divided into $K$ markets, and each firm operates
in one market off-platform, so each off-platform market has $N=J/K$ firms
(and assume $N$ is an integer).

Now, the posted price impacts the off-platform market slightly differently.
In particular, by setting a posted price $p$ when the competitors in the
off-platform market set price $p^{\prime }$, the firm wins an off-platform
consumer if and only if $v-p\geq v^{\prime }-p^{\prime }$, where $v^{\prime }$ is the value of the best competitor. Hence, the firm's profit off the platform is given by 
\begin{equation}
\int_{p}^{\bv} pF^{N-1}(v-p+p^{\prime })dF(v).\label{persal}
\end{equation}%
In the absence of the platform, the symmetric equilibria of the game with payoffs \eqref{persal} yield the 
oligopoly prices, which we denote by $p_{O}^{K}$. These are the  prices the firms would charge if there were $K$ segmented markets. The basic model considered $K=J$ segmented markets, so that each firm was acting
as a monopolist in its market; thus for $K=J,\ p_{O}^{J}=p_{M}$. We now
investigate how a more competitive off-platform market affects the behavior
in the bidding and managed campaign mechanism. 

To this end, note that the on-platform profit terms in our previous analysis are not affected by changes in the off-platform market as suggested above.  
We now consider the symmetric equilibrium off-platform prices with competition in $K$ market
segments, both under  data-augmented bidding $(p_{B}^{K})$ and under the fully optimal managed campaign $(p_{V}^{K})$. We obtain the following comparison.

\begin{proposition}[Off-Platform Competition]
\label{prop:off_competition} In the highest-price symmetric equilibria with off-platform competition in $K$ markets, the off-platform posted prices satisfy
\begin{equation*}
p_{O}^{K}\leq p_{B}^{K}\leq p_{V}^{K}.
\end{equation*}
\end{proposition}

Thus, the equilibrium ordering of the off-platform prices, and the
corresponding welfare results are invariant to the structure of competition in
the off-platform markets. In particular, the above ordering is same as
in Theorem \ref{thm:managed_prices_highest_variant} where we observed: $%
p_{M}\leq p_{B}\leq p_{V}$. 

\subsection{Platform Revenue Models}

In our modeling of the managed campaign, the platform requests an up-front participation fee, the advertising budget. This aligns with the
practice of managed campaigns on pure advertising platforms such as google,
Facebook or Tiktok. These platforms match advertisers and consumers, but do
not charge any transaction fees. By contrast, shopping or retail platform
such as amazon or Instacart typically receive revenue from a mixture of
advertising and sales commissions. One might thus wonder whether there are
multiple, and payoff equivalent mechanisms that could all attain the same
total revenue. Here, we shall focus one such alternative in which the
platform is charging a constant transaction fee $t_{j}$ to each firm and
does not impose a fixed payment $T_{j}$. We then show that for modest
transaction fees, the firms' incentives for setting prices off the platform
are exactly as in the managed campaign we discussed in Theorem \ref%
{thm:managed_eq_best_value}

\begin{proposition}[Transaction Fees]

\label{prop:transaction_fees} Suppose the platform proposes a transaction
fee $t_{j}$ for each sale of firm $j$ on the platform. There exists a $\bar{t%
}>0$ such that if the fee satisfies $t_{j}\leq \bar{t}$, then the off platform
prices $\bar{p}$ are the same as in the best-value managed campaign of Theorem \ref{thm:managed_eq_best_value}.

\end{proposition}

In the managed campaign the platform receives all the revenue from the sponsored product placement through the fixed fees. By contrast, the platform does not receive a revenue share or a commission on the sales realized on the platform. The revenue from the sales accrues directly to the advertiser. In the current revenue model, we replaced the ex ante advertising
budget with a constant sales (or referral) fee. The key insight is that the referral fee does not influence the marginal incentives of the firm's pricing decision; hence, provided the referral fee is not too large and the firm is still willing to participate, the pricing decision off platform remains unchanged. Thus, there are revenue models that are revenue equivalent to the managed campaign mechanism.

\section{Conclusion}

Many digital platforms such as Google, Meta, Amazon, and TikTok generate
revenue through advertising by placing ads or sponsored slots on their own
and partner websites. These platforms use a combination of manual and automated bidding mechanisms to select  valuable advertisements to display to each user and to set prices for these ads. The
platform's knowledge about the match value between consumers and products is
critical to the success of both mechanisms. This knowledge helps generate
the most competitive bids from advertisers and supports 
clicks and other forms of user engagement with the platform. 

We have proposed an
integrated model that considers how auction mechanisms and data availability jointly determine match formation and surplus  extraction both on and off large digital platforms. The auction mechanisms employed by  the platform have substantial implications for
product prices. On the platform, the data made available to the advertisers
allows for efficient matching, yet most of the surplus accrues to the
platform. Off the platform, advertisers raise prices  to gain a competitive edge on the platform. 

The cross-channel distortions become more pronounced the more tools the platform has at its disposal, relative to  the traditional (generalized) second price auctions. Indeed, we have shown that the higher costs of online advertising under a more extractive mechanism are passed on to consumers by means of higher product prices \emph{off} rather than \emph{on} the platform. These results  suggest the need for further
analysis of how algorithmic bidding impacts competition and welfare in all markets, particularly off the large digital platforms. 

\pagebreak

\appendix

\section{Appendix}


The proof of Proposition \ref{prop:bidding_efficiency} is given in the text.

\paragraph{Proof of Theorem \ref{thm:bidding_eq}}
(1.) This part follows directly from Proposition \ref{prop:bidding_efficiency}. 

(2.) Because we are looking for symmetric equilibria,  suppose all the other firms post price $p'$ and consider the best response problem of a single firm:
\begin{equation}\label{eqn:bidding_profit_expr}
\max_p \left\{ \frac{1-\lambda}{J}\, p(1 - F(p)) + \lambda\Omega(p; p') \right\},
\end{equation}
where
\[ \Omega(p ; p') = \int_{\underline{v}}^{\bar{v}} \int_{\underline{v}}^v ( \min(v - \max(v' - p', 0) , p) - (\min(v' - \max(v-p,0), p'))_+) dF^{J-1}(v') dF(v).  \]
This term denotes the expected profit from on-platform consumers that a firm would expect to make by setting a posted price at $p$ when all other firms set a posted price $p'$. The term integrates over $v' = \max_{j \neq i} v_j$, which is the highest value the consumer has for any other firm besides $i$. Since the firm must concede utility $\max(v' - p', 0)$ to the threat of the on-platform consumer going to the competitor, the firm setting price $p$ will bid $\min(v - \max(v' - p', 0), p)$. The highest competitor bids $(\min(v' - \max(v-p,0), p'))_+$.

With some casework, we can show that the expected on-platform profit satisfies \begin{equation}\label{eqn:omega_definition}
    \Omega(p;p') =  \int_{\underline{v}}^{\bar{v}} \int_{\underline{v}}^v \min(v - v', p) \ dF^{J-1}(v') \ dF(v), \text{ for all $p'$}.
\end{equation}
Expression \eqref{eqn:omega_definition} shows that  each firm's profit in the second-price auction cannot exceed the difference between its own value and the next highest value. Furthermore, this difference is capped by the firm's own posted price $p$ due to the showrooming constraint. Additionally, this expression  is independent of $p'$, so we suppress the dependence on $p'$ in the notation and write $\Omega(p)$ instead.

To characterize  the symmetric equilibria, we compute the derivative of $\Omega$ with respect to $p$. Straightforward algebra yields the following expression:
\begin{equation}\label{eqn:onplatform_profit_deriv}
\Omega'(p)=\int_p^{\bar{v}} F^{J-1}(v - p)\ dF(v).
\end{equation}
Finally, we can write out the first-order condition for profit maximization using  \eqref{eqn:onplatform_profit_deriv}:
\begin{equation*}
     \frac{1 - \lambda}{J}\left( 1 - F(p) - pf(p) \right) + \lambda \left( \int_p^{\bar{v}} F^{J-1}(v - p) dF(v) \right) = 0.
\end{equation*}

 (3.)-(4.) These results follow from setting  $\bar{p}_j=p_B$ in the buyer's outside option \eqref{eq:outside} and recalling from the proof of Proposition \ref{prop:bidding_efficiency} that firm $j$ bids $b_j(v,\bar{p})=(v_j-\underline{u}(v,\bar{p}))_+$. \hfill \qedsymbol

\paragraph{Proof of Proposition \ref{prop:bidding_price_larger}} (1.) The equilibrium price maximizes the profit function
\begin{equation*}
    \frac{1-\lambda}{J}p(1 - F(p)) + \lambda \Omega(p), 
\end{equation*}
where $\Omega(p)$ is given in \eqref{eqn:omega_definition}. Equivalently, $p_B$ maximizes the rescaled profit function
\begin{equation*}
    \frac{1}{J}p(1 - F(p)) + \frac{\lambda}{1-\lambda} \int_{\underline{v}}^{\bar{v}} \int_{\underline{v}}^v \min(v - v', p) \ dF^{J-1}(v') \ dF(v). 
\end{equation*}
Because the second term is strictly increasing in $p$ and $\lambda$, and it is multiplicatively separable, this function is  supermodular in $(p,\lambda)$; hence, by Topkis's theorem \citep{topk78}, the
profit-maximizing posted price is nondecreasing in $\lambda$. 

(2.)  The expected consumer surplus of an off-platform consumer and the expected welfare per consumer off-platform are given by 
\[ CS_{\off}(p) = \int_p^{\bar{v}} (v-p) \ d F(v)\quad\text{and}\quad W_\off(p) = \int_p^{\bar{v}} v \ dF(v) , \]
respectively. Both quantities are strictly decreasing in $p$, and hence also in  $\lambda$.\hfill\qed

\paragraph{Proof of Theorem \ref{thm:managed_eq_best_value}}
(1.) This follows from the definition of the best-value pricing  and efficient steering policies when all firms post an identical price $\bar{p}_j\equiv p_V$.

(2.) To characterize the symmetric equilibrium prices, suppose first the firm sets a price $p < p'$. In this case, under the best-value pricing policy, the firm will not poach any on-platform consumer for which it is not the high-value firm. Thus, the firm collects $\min(v,p)$ on all consumers for which it is the high-value firm. The firm's profit is given by
\begin{equation}
\Pi(p,p')= \frac{1-\lambda}{J}p(1-F(p)) + \lambda \left( \int_{\underline{v}}^p v F^{J-1}(v) \ dF(v) + \int_p^{\bar{v}} p F^{J-1}(v) \ dF(v) \right).\label{pilow}
 \end{equation}
The derivative with respect to $p$ is given by
\begin{align}
    &\frac{1-\lambda}{J}(1-F(p) - pf(p)) + \lambda \left(p F^{J-1}(p)f(p) - pF^{J-1}(p)f(p) + \int_p^{\bv} F^{J-1}(v) \ dF(v) \right)\notag \\
    =&\frac{1-\lambda}{J}(1-F(p) - pf(p)) + \lambda \left( \int_p^{\bv} F^{J-1}(v) \ dF(v) \right).\label{eq:foc_managed_variant_left}
\end{align}
Now, suppose the firm posts a  price $p > p'$. The firm's profit function is given by
\begin{align}
 \Pi(p,p')&=\frac{1-\lambda}{J}p(1-F(p)) + \lambda  \begin{pmatrix} \int_{\uv}^{p'} v F^{J-1}(v) \ dF(v) + \int_{p'}^p \int_{\uv}^{p'} v \ dF^{J-1}(v') \ dF(v)  \\ 
 +\int_{p'}^p \int_{p'}^v (v - (v' - p')) \ dF^{J-1}(v') \ dF(v)\\
 + \int^{\bv}_{p} \int_{\uv}^{p' + (v-p)} p \ dF^{J-1}(v') \ dF(v)   \\
  +\int_{p}^{\bv} \int_{p' + (v-p)}^v (v - (v'-p')) \ dF^{J-1}(v') \ dF(v)
 \end{pmatrix}.\label{pihigh}
\end{align}
With some algebra, one can show that the derivative of this expression with respect to $p$ is
\begin{eqnarray}
 \frac{1-\lambda}{J}(1-F(p) - pf(p)) + \lambda \left(\int_{p}^{\bv}  F^{J-1}(p' + v - p) \ dF(v)
 \right). \label{eq:foc_managed_variant_right}
\end{eqnarray}
Comparing (\ref{eq:foc_managed_variant_left}) and (\ref{eq:foc_managed_variant_right}), the derivative matches from the left and right at $p = p'$, and so the best-response function is continuously differentiable at $p'$ with derivative
\[\frac{1 - \lambda }{J}(1 - F(p) - pf(p)) +  \lambda \int_p^{\bv} F^{J-1}(v) \ dF(v). \]
This expression is strictly positive at $p=\uv$ and strictly negative at $p=\bv$. Therefore, a necessary condition for a symmetric equilibrium is given by the first-order condition \eqref{eq:managed_price_best_value} that sets this derivative to zero.  Furthermore, if a single price satisfies \eqref{eq:managed_price_best_value}, then the symmetric equilibrium is unique. \hfill \qedsymbol

\paragraph{Proof of Theorem \ref{thm:best_value_max}}
The vertically integrated platform, which controls all the firms' prices, sets identical prices $\bar{p}_j=p$ to maximize
    \[ \Pi_C(p)\triangleq \left\{ (1-\lambda)p(1 - F(p)) + \lambda \int_{\uv}^{\bv} \min(v, p) \ dF^{J}(v)  \right\}. \]
    Now compare this problem to a firm's best reply to a common price $p'$ posted by its competitors: the firm's profit $\Pi(p,p')$ in \eqref{pilow} coincides with $\Pi_C(p)/J$ on $p\in[\uv,p']$; and the firm's profit $\Pi(p,p')$ in \eqref{pihigh} satisfies $\Pi(p,p')<\Pi_C(p)/J$ on  $[p',\bv]$. Now let $p\in\arg\max_p \Pi_C(p)$ denote a solution to the vertically integrated platform's problem. By construction, we have that $p\in\arg\max_p \Pi(p,p^*)$. Therefore, any $p^*$ is a symmetric equilibrium of the best-value pricing managed campaign, and the largest $p^*$ is also the highest symmetric equilibrium price of the managed campaign. Finally, any $p^*$ attains the integrated producer surplus level.
    
    
    
To show the optimality of best-value pricing and efficient steering, consider the platform revenue, which equals the joint surplus generated by this mechanism, net of the firms' outside option value defined in \eqref{eqn:outside}. This level of the outside option is a lower bound on the profit of a firm that  refuses to participate in any mechanism. Therefore,  the
managed campaign we are considering maximizes the joint surplus of the platform and firms, and it concedes the smallest possible surplus to the firms. It follows that the best-value pricing campaign maximizes the platform's revenue.\hfill\qed

\paragraph{Proof of Theorem \ref{thm:managed_prices_highest_variant}}
We can nest the optimal pricing problem across the three problems (monopoly, auction, campaign) with a parameter $\gamma$. Consider the choice of posted price in each of the three models. Define the auxiliary profit  function
\[ \Pi(p, \gamma)\triangleq\frac{1-\lambda}{J}\left( 1 - F(p) \right) p  + \lambda \left( \int_{\uv}^{\bv} \int_{\uv}^v \min(v - \gamma v',p) dF^{J-1}(v') dF(v)\right). \]
In the data-augmented auctions, each firm's profit function is given by $\Pi(p, 1)$.
The profit function of the vertically integrated firm (which yields the equilibrium price in the best-value pricing managed campaign by Theorem \ref{thm:best_value_max}) is given by $\Pi(p, 0)$. It is straightforward to verify that $\Pi$ is submodular in $(p, \gamma)$: 
\[ \frac{\partial^2 \Pi(p,\gamma)}{\partial \gamma \partial p} = -\int_p^{\bv} (v-p) (J-1)F^{J-2}(v-p)f(v-p)f(v) \ dv < 0.\]
Thus, by Topkis's theorem, the largest maximizer of $\max_p \Pi(p,\gamma)$ is nonincreasing in $\gamma$. Since $p_V = \max\{\arg\max_p \Pi(p,0)\}$ and $p_B = \max\{\arg\max_p \Pi(p,1)\}$, it follows that $p_V \ge p_B$. Note that this also implies $p_V \ge p_B \ge p_M$ by Proposition \ref{prop:bidding_price_larger}. Finally, note that  the matching of consumers to firms is identical across the two mechanism. Thus, the comparison of total surplus and consumer surplus is entirely driven by the posted price. Because both surplus levels are decreasing in $p$, the welfare comparative statics follow.\hfill\qed

\paragraph{Proof of Proposition \ref{prop:passthrough}}
Note that the joint profit outcome of the platform and firms in both bidding and the sophisticated managed campaign takes the form:
\[ \Pi_J(p) = (1-\lambda)p(1 - F(p)) + \lambda \int_{\uv}^{\bv} \min(v,p) dF^{J}(v),  \]
and the bidding profit outcome is $\Pi_J(p_B)$ while the sophisticated managed campaign outcome is $\Pi_J(p_V)$. Fixing the bargaining power of the platform, the total transfer charged to all firms is $\Pi_J(p)/J - \Pi_O$, and so we can compute $\eta$ as 
\[ \frac{p_V - p_B}{(1-\lambda)/J\left[p_V(1 - F(p_V)) - p_B(1 - F(p_B)) \right] + \lambda/J \int_{p_B}^{\bv} (\min(v,p_V) - p_B) F^{J-1}(v)dF(v)  }, \]
\[ = \frac{J}{(1-\lambda)\left[1 - \frac{p_VF(p_V)- p_BF(p_B)}{p_V - p_B} \right] + \lambda \int_{p_B}^{p_V}\frac{v - p_B}{p_V - p_B}dF^J(v)  + \frac{\lambda}{J} \int_{p_V}^{\bv} dF^J(v)  }. \]
We claim that the denominator of the above expression is less than 1. To see this, note the bracketed term is less than 1 because $p_V > p_B$, so the first term is at most $1-\lambda$. The second term is integrated for $v \in [p_B, p_V]$, so $0 \le v - p_B \le p_V - p_B$. Hence the sum of the last two terms is bounded above by the integral $\lambda \int_{p_B}^{\bv} dF^J(V) = \lambda (1 - F^J(p_B)) < \lambda$. Thus, the denominator is at most $1$, with strict inequality for $\lambda > 0$, and so $\eta > J$. \hfill\qed

\paragraph{Proof of Theorem \ref{thm:indep_managed_campaign}}
Suppose, for sake of contradiction, that in the optimal independent managed campaign, $p_I > p_V$, and the platform chose some pricing policy $p^*(v)$. Consider the best-response problem of a firm in the independent managed campaign. The best-response profit function for all $p < p_I$ is 
{\small \[ \Pi_I(p)= \frac{1-\lambda}{J}p(1 - F(p)) + \lambda \begin{pmatrix}  \int_{\uv}^{\bv} \min(p^*(v),p)  F^{J-1}(v) \ dF(v)
\\ +\int_{p} p \left[ \int_v^{\bv} \mathbbm{1}[v_j - p^*(v) < v - p]dF^{J-1}(v_j) \right] \ dF(v) \end{pmatrix}. \]}
We can split this into two components. Denote the first component 
\[ \Pi_A(p) \triangleq  \frac{1-\lambda}{J}p(1 - F(p)) + \lambda \int_{\uv}^{\bv} \min(p^*(v),p)  F^{J-1}(v) \ dF(v), \]
which is the profit from off-platform sales and on-platform sales in the segment that of on-platform consumers that prefer firm $i$, and the second component 
\[ \Pi_B(p) \triangleq \lambda \int_{p}^{\bv} p \left[ \int_v^{\bv} \mathbbm{1}[v_j - p^*(v) < v - p]dF^{J-1}(v_j) \right] \ dF(v),  \]
which is the segment of consumers who prefer some other $j \neq i$ but are poached by $i$. Note that by construction, $\Pi_I(p) = \Pi_A(p) + \Pi_B(p)$. By the contradiction supposition, $p_I > p_V$ was the outcome of the optimal independent managed campaign. Note that $\Pi_B(p_I) = 0$, since no consumers can be poached when $i$ sets the same posted price as all other firms. Since $\Pi_B \ge 0$ for all $p < p_I$, it follows that $\Pi_B(p_V) - \Pi_B(p_I) \ge 0$. 

We now claim that $\Pi_A(p_V) - \Pi_A(p_I) > 0$. Consider $\Pi_A'$. Let $\phi^*$ denote the preimage of $p^*$: $\phi^*(p) = \{ v \in V | p^*(v) \ge p, \ i = \arg\max v_i\}$. Then 
\[ \Pi_A'(p) = \frac{1-\lambda}{J}(1 - F(p) - pf(p)) + \lambda \mu(\phi^*(p)), \]
where $\mu$ is the probability measure over the type space. 
Note that since it is without loss for the platform pricing policy to never set a price larger than the consumer's value, $\phi^*(p) \subseteq  \{ v \in V | v \ge p,  i = \arg\max v_i\} =:  \bar{v}(p)$. 
Thus, $\mu(\phi^*(p)) \le \mu(\bar{v}(p)) = \int_p^{\bv} F^{J-1}(v) d F(v) $. Substituting this in, we get
\begin{align*} \Pi_A'(p) &= \frac{1-\lambda}{J}(1 - F(p) - pf(p)) + \lambda \mu(\phi^*(p)) \\
&\le \frac{1-\lambda}{J}(1 - F(p) - pf(p)) + \lambda \int_p^{\bv} F^{J-1}(v) d F(v)  = \Pi_V'(p), \end{align*}
where $\Pi_V$ was the joint vertical integration profit from the sophisticated managed campaign. Therefore, we have $\Pi_A'(p) \le \Pi_V'(p)$, so 
\[ \Pi_A(p_V) - \Pi_A(p_I) = \int_{p_V}^{p_I} - \Pi_A'(p) dp \ge \int_{p_V}^{p_I} - \Pi_V'(p) dp = \Pi_V(p_V) - \Pi_V(p_I) > 0, \]
where the last inequality follows because $p_V$ by definition is the largest maximizer of $\Pi_V$. But this implies that $\Pi_A(p_V) - \Pi_A(p_I) > 0$ and $\Pi_B(p_V) - \Pi_B(p_I) \ge 0$, and hence $\Pi_I(p_V) - \Pi_I(p_I) > 0$, which contradicts our supposition that the individual best response was to set posted price $p_I$. Thus, it cannot be that $p_I > p_V$.

To show that we cannot have equality, suppose for sake of contradiction that $p_I = p_V$, and consider $\Pi_I'(p_V)$. By the same argument above, we have that $\Pi_A'(p_V) \le \Pi_V'(p_V) $. If $\Pi_A'(p_V) > \Pi_V'(p_V)$ strictly, then since $\Pi_B(p_V) = 0$ and $\Pi_B(p_V) \ge 0$, there exists an $\epsilon$ such that $\Pi_I'(p_V - \epsilon) > \Pi_I'(p_V)$, a contradiction. If $\Pi_A'(p_V) = \Pi_V'(p_V)$, then it must be that $\phi^*(p) = \bar{v}(p)$, so every consumer with value at least $p$ sees a price at least $p$. But this implies that poaching can happen; more precisely, the left-derivative $(\Pi_B)'_{-}(p_v) = - \int_p^{\bv} p dF^{J-1}(v) dF(v)$. Since the left derivative of $\Pi_B$ is strictly negative, and $\Pi_A'(p_V) = \Pi_V'(p_V)$, for small enough $\epsilon$ then we must have $\Pi_I(p_V - \epsilon) > \Pi_I(p_V)$, a contradiction again. Hence we cannot have $p_I = p_V$. Thus $p_I < p_V$. Since the welfare and off-platform surplus are decreasing in posted price for $p \ge p_M$, the welfare comparative statics follow. \hfill \qedsymbol

\paragraph{Proof of Proposition \ref{prop:partial_privacy}}
    First, consider the problem of a vertically integrated platform facing the cohort-privacy constraint:
    \[  \Pi_P(p) := \max_{p' \le p} \frac{1-\lambda}{J}p(1- F(p)) + \frac{\lambda }{J} p' (1 - F^{J}(p')). \]
    Let $p_P$ denote the largest maximizer of $\Pi_P$. We can write the Lagrangian:
    \[  L(p,\mu) := \frac{1-\lambda}{J}p(1- F(p)) + \frac{\lambda }{J} p' (1 - F^{J}(p')) + \mu (p' - p), \]
    where $\mu$ is the multiplier associated with the showrooming constraint. Because $F^J$ satisfies the monotone-likelihood ratio with respect to $F$, the unconstrained maximum must have $p' > p$; hence, the showrooming constraint binds,  $p' = p$, and the optimal $p_P$ satisfies: 
    \begin{equation*}
     \frac{1-\lambda}{J}(1 - F(p) - pf(p)) + \frac{\lambda}{J}(1 - F^J(p) - JpF^{J-1}(p)f(p)) = 0.      
    \end{equation*}
    Note that $\Pi_P(p_P)$ can be rewritten as
    \[ \Pi_P(p_P) = \frac{1}{J}p_P(1 - ((1-\lambda)F(p_P) + \lambda F^{J}(p_P) ). \]
    The necessary first-order condition for optimality thus requires $p_P$ to be equal to the inverse hazard rate of the distribution $(1-\lambda)F + \lambda F^J$; since $(1-\lambda)F + \lambda F^J$ satisfies the monotone likelihood ratio property with respect to $F$, it follows that $p_P \ge p_M$. Similarly, since $F^J$ satisfies the monotone likelihood ratio property with respect to $(1-\lambda)F + \lambda F^J$, it follows that the maximizer $p_J := \arg\max_p p (1 - F^{J}(p))$ is larger than $p_P$.  
    
    We now show that the platform can implement an equilibrium where all firms post $p_P$. Consider the platform pricing policy, which sets the price for cohort $i$ as the minimum posted price of firms $j \neq i$. The best-response profit of firm $i$ when all other firms set price $p_P$ is
    \[ \begin{cases} 
      \frac{1-\lambda}{J}p(1- F(p)) + \frac{\lambda }{J} p (1 - F^{J}(p)) & p < p_P, \\
      \frac{1-\lambda}{J}p(1- F(p)) + \frac{\lambda }{J} p_P (1 - F^{J}(p_P)) & p \ge p_P.
    \end{cases} \]
    Suppose, towards a contradiction, that $p_P$ was not a maximizer, and some other $p^*$ was the maximizer. If $p^* < p$, that implies that $\Pi_P(p^*) > \Pi_P(p_P)$, a contradiction of the definition of $p_P$. If $p^* > p$, then there exists $p' = p_P$, $p = p^*$ such that 
    \[ \frac{1-\lambda}{J}p(1- F(p)) + \frac{\lambda }{J} p' (1 - F^{J}(p')) > \Pi_P(p_P), \]
    again a contradiction of the optimality of $p_P$. Hence, the best-response of each firm is also to set price $p_P$. Since this attains the vertical integration profit, it is optimal for the platform to set such a pricing policy.

    It remains to show that $p_V \ge p_P$. Note that $p_P$ maximizes 
    \[  \Pi_P(p) = \frac{1-\lambda}{J}p(1- F(p)) + \frac{\lambda }{J} p (1 - F^{J}(p)),\vspace{-.1in} \]
    and $p_V$ maximizes\vspace{-.1in}
    \[ \Pi_V(p) =\frac{1-\lambda}{J}p(1- F(p)) + \frac{\lambda }{J} \left( \int^p_{\uv} v dF^J(v) + \int_p^{\bv} p dF^{J}(v) \right).  \]
    To complete the argument, define the auxiliary profit function again:
    \[ \Pi(p, \gamma) = \frac{1-\lambda}{J}p(1- F(p)) + \frac{\lambda }{J} \left( \gamma \int^p_{\uv}  v dF^J(v) + \int_p^{\bv}  p dF^{J}(v) \right), \]
    and note that $\Pi(\cdot,0) = \Pi_P$ and $\Pi(\cdot,1) = \Pi_V$. Finally, note that 
    \[ \frac{\partial^2 \Pi(p,\gamma)}{\partial \gamma \partial p} =  p\, dF^{J}(p) > 0.\]
    By Topkis's theorem, the maximizer of $\Pi(\cdot, \gamma)$ is nondecreasing in $\gamma$; hence $p_P \le p_V$.\hfill\qed

\paragraph{Proof of Proposition \ref{prop:off_competition}}
We consider the benchmark without a platform (denoted by $O$ for oligopoly), the data-augmented bidding $B$, and the best-value pricing campaign $V$. 
To compare  these three cases, consider an auxiliary game with payoffs
\begin{align}\label{eqn:aux}
    \Pi_{\off}(p,\bar{p};\lambda, \gamma) =&  \frac{1}{M} p \int_{p}^{\bv} F^{N-1}(v - (p - \bar{p})) f(v) \ dv \notag\\
    & +\frac{\lambda}{1-\lambda} \int_{\uv}^{\bv} \int_{\uv}^v \min(v- \gamma v',p) dF^{J-1}(v') dF(v).
\end{align}
The game with $\lambda=0$ describes the case without a platform; the game with $\gamma=1$ describes the data-augmented auctions; and a similar argument as in the proof of Theorem \ref{thm:best_value_max} establishes that every symmetric equilibrium in the game with $\gamma=0$ is also an equilibrium under best-value pricing. We can then define the following prices:\vspace{-.1in}

\[ p^K_O \in \arg\max_p \Pi_{\off}(p,p^K_O;0,1),  \]
\[ p^K_B \in \arg\max_p \Pi_{\off}(p,p^K_B;\lambda,1), \]
\[ p^K_V \in \arg\max_p \Pi_{\off}(p,p^K_V;\lambda,0). \]

To compare $p^K_O$ and $p^K_B$, consider the  best reply functions $p^*(\bar{p})$ in the game with $\lambda=0$ and in the game with $\lambda>0$ and $\gamma=1$.  The payoff function $\Pi_\off$ in \eqref{eqn:aux} has increasing differences in $(p,\lambda)$ for any $\bar{p}$ when $\gamma=1$. Indeed, we have
\[ \frac{\partial^2\Pi_\off(p,\bar{p};\lambda,1)}{\partial p \partial \lambda } = \frac{1}{(1-\lambda)^2} \int_{p}^{\bv} F^{J-1}(v - p) dF(v) > 0. \] 
By Topkis's theorem, a higher $\lambda$ increases the best reply $p^*$ to any price $\bar{p}$. Furthermore, the best response function  satisfies $p^*(\bv)<\bv$ in all three cases. Therefore,  in the highest-price symmetric equilibrium, the best-reply function crosses the line $p^*=\bar{p}$ from above, and an upward shift in the best replies raises the highest symmetric equilibrium price, i.e., $p^K_B>p^K_O$.

Similarly, the payoff function \eqref{eqn:aux} has decreasing differences in $(p,\gamma)$ for all $\bar{p}$ and  $\lambda>0$:
\[ \frac{\partial^2 \Pi_\off(p,\bar{p};\lambda,\gamma)}{\partial p \partial \gamma } = -\frac{\lambda}{1-\lambda}\int_p^{\bv} (v-p) (J-1)F^{J-2}(v-p)f(v-p)f(v) \ dv < 0. \]
By Topkis's theorem, the best replies in the game with $\gamma=0$ (i.e., the managed campaign) are pointwise larger than in the game with $\gamma=1$ (i.e., the data-augmented auctions). Furthermore, the best replies when $\gamma=1$  satisfy $p(\bv)<\bv$, and hence the largest symmetric equilibrium price increases as best replies shift up, i.e.,  $p^K_V>p^K_B$. Finally, because every equilibrium in the auxiliary game is also an equilibrium under the best-value pricing campaign,  the largest symmetric equilibrium in the original game is strictly larger than $p^K_B$.\hfill\qed

\paragraph{Proof of Proposition \ref{prop:transaction_fees}}
Under best-value pricing,  each firm $j$ sells to the $1/J$ fraction of consumers who like its product best, since by construction  
$v_j - p(v,a,\bar{p}) \ge v_k - \bar{p}_k$. Thus, the pricing decisions of the firms are the same as in Theorem \ref{thm:managed_eq_best_value}. The platform could then charge transaction fees instead of an upfront budget with the same outcome, provided the transaction fee satisfies the participation constraint of the firms: equivalently $\lambda t_j/J \le T_j$, where $T_j$ is what the platform charged in the original managed campaign model.\hfill\qed

\newpage

\bibliographystyle{econometrica}
\bibliography{general,ale2,nick}

\newpage

\section{Online Appendix}
\subsection{Comparative Statics}\label{app:comparative}

\subsubsection{Data-Augmented Bidding}

We present additional welfare comparative statics for the data-augmented bidding model. We will first define several  welfare objects of interest, as functions of the posted price $p$. The expected consumer surplus of an off-platform consumer is the difference between their willingness-to-pay and price, which is 
\[ CS_{\off}(p) = \int_p^{\bv} (v-p) \ d F(v). \]
Note that only off-platform consumers with value above $p$ receive a positive surplus, since consumers with low values do not buy.
The expected consumer surplus of an on-platform consumer is:
\[ CS_{\on}(p) = \int_p^{\bv} (v-p) dF^J(v).  \]
Again, only consumers with sufficiently high values receive positive surplus, since low-value consumers are price discriminated against.
Because $F^J$ describes the distribution of the highest-order statistic, the expected welfare of an on-platform consumer is always larger than an off-platform consumer's.

Similarly, the expected welfare per consumer off-platform is given by 
\[ W_\off(p) = \int_p^{\bv} v \ dF(v), \]

Note that we focus on off-platform welfare because under our model, the platform has full information and sales are always made, there is full welfare on-platform.

\begin{proposition}[Comparative Statics]\label{prop:bidding_comp_statics}\strut
Suppose $J > -1/(\ln F(\bv-p_M))$. The following comparative statics hold:
    \begin{enumerate}
        \item The equilibrium posted price with data-augmented bidding $p_B$ is decreasing in $J$. 
        \item The expected consumer surplus of on-platform and off-platform consumers is increasing in $J$.
        \item Off-platform welfare per consumer $W_\off$ is increasing in $J$.  
    \end{enumerate}
\end{proposition}

\begin{proof}
    

By rescaling the best-response profit condition solved by the bidding problem, the bidding price is a maximizer of
\[ \Pi(p) = (1-\lambda)p(1-F(p)) + \lambda \left(\int_{\uv}^{\bv}\int_{\uv}^v \min(v - v', p) J dF^{J-1}(v') dF(v) \right) . \]
Note that the cross-partial derivative with respect to $J, p$ is 
\begin{eqnarray*}
\frac{\partial^2 \Pi}{\partial J \partial p }&=&\int_{p}^{\bv} (1 + J \ln F(v - p)) F^{J-1}(v-p) dF(v).
\end{eqnarray*}
Since by assumption $ 1 < - J \ln F(\bar{v} - p_M)$, $1 + J \ln F(v-p) \le 1 + J \ln F(\bar{v} - p_M) < 0$ for $p \ge p_M$. Hence this derivative is negative, so by Topkis's theorem, it thus follows that $p_B$ must be decreasing in $J$.

The second statement follows from the first and the fact that $CS_\off$ and $CS_\on$ are both decreasing in $p$. Since $CS_\on$ is increasing in $J$, it follows that both quantities are increasing in $J$. Finally, the total off-platform welfare per consumer is  decreasing in $p$, and has no other $\lambda$ or $J$ dependence, so $W_\off$ is increasing in $J$. 
\end{proof}

We illustrate many of these comparative statics with a simple example. 

\paragraph{Example} Consider the setting where the distribution $F$ is uniform on $[0,1]$. Note that in this setting, since the distribution is uniform, the monopoly price is $p_M = 0.5$. We plot the equilibrium posted prices, total firm profit, and consumer surplus resulting from data-augmented bidding for $J = 3, 5, 7$ in Figure \ref{fig:bidding_prices}. As shown in Proposition \ref{prop:bidding_comp_statics}, for any $J$, the prices are increasing in $\lambda$.  
\begin{figure}[ht]
    \centering
    \includegraphics[width=0.6\textwidth]{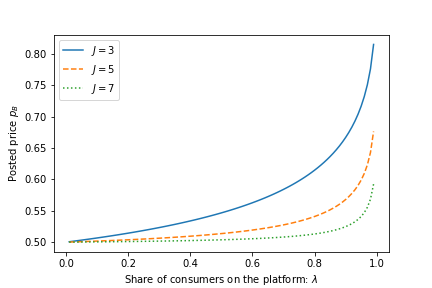}
    \caption{Posted prices as a function of $\lambda$. Results are plotted for $J= 3 ,5, 7$.}
    \label{fig:bidding_prices}
\end{figure}

\begin{figure}[ht]
    \centering
     \begin{subfigure}[b]{0.49\textwidth}
         \centering
         \includegraphics[width=\textwidth]{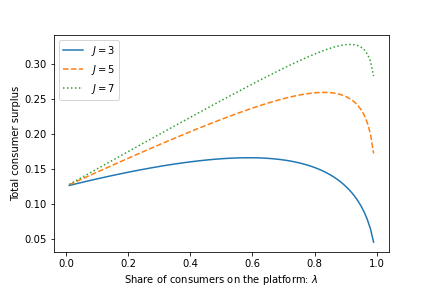}
         \caption{Total Consumer Surplus}
         \label{fig:bidding_cons_surplus}
     \end{subfigure}
     \hfill
     \begin{subfigure}[b]{0.49\textwidth}
         \centering
         \includegraphics[width=\textwidth]{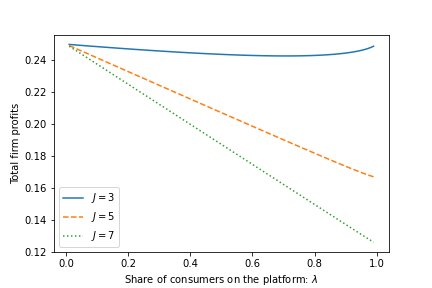}
         \caption{Total Firm profit}
         \label{fig:bidding_firm_profit}
     \end{subfigure}
     \\
     \begin{subfigure}[b]{0.49\textwidth}
         \centering
         \includegraphics[width=\textwidth]{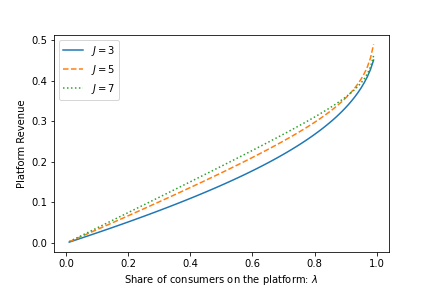}
         \caption{Platform Revenue}
         \label{fig:bidding_platform_revenue}
     \end{subfigure}
     \hfill
     \begin{subfigure}[b]{0.49\textwidth}
         \centering
         \includegraphics[width=\textwidth]{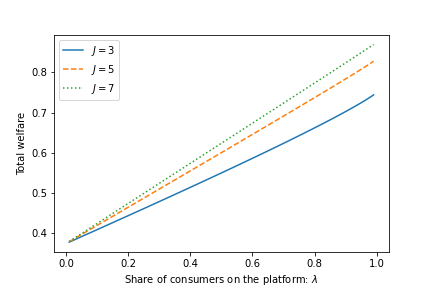}
         \caption{Total Welfare}
         \label{fig:bidding_welfare}
     \end{subfigure}
    \caption{Consumer surplus, firm profit, platform revenue, and welfare with data-augmented bidding as a function of the share of consumers on the platform, $\lambda$. Results are plotted for $J = 3, 5, 7$.}
    \label{fig:bidding_figure}
\end{figure}

Figure \ref{fig:bidding_cons_surplus} depicts the consumer surplus as a function of $\lambda$. We note that total consumer surplus is increasing in $\lambda$. Initially, the welfare gains from moving consumers from being loyal to shopping over all firms dominate (moving consumers from welfare level $CS_\off$ to $CS_\on$) but as the platform becomes too large, the increasing ability to price discriminate on the platform dominates and consumers lose welfare. Hence, total consumer surplus is nonmonotone in $\lambda$. 

Figure \ref{fig:bidding_firm_profit} depicts firm profit as a function of $\lambda$. Here, firm profit for $J = 3$ are nonmonotone. As mentioned in Proposition \ref{prop:bidding_comp_statics}, the profit per consumer off-platform is decreasing in $\lambda$ and the profit per consumer on-platform is increasing in $\lambda$, and so the overall effect on total profit depends on which force dominates.

Figure \ref{fig:bidding_platform_revenue} depicts the platform revenues as a function of $\lambda$. As expected, platform revenues are increasing in $\lambda$. However, the interesting feature of this example in platform revenue is that for very large platforms, $\lambda$ close to 1, the platform revenue can be nonmonotone in $J$, the number of firms. The two contrasting forces here are that with more firms, the expected value of second-highest bids will be higher, which would suggest that platform revenue should be increasing in $J$. However, with more firms, as shown in Figure \ref{fig:bidding_prices}, posted prices can be pushed down, thus reducing the price-discriminating ability of the firms on-platform and pushing down the platform revenues. 

Figure \ref{fig:bidding_welfare} shows that total welfare is increasing in both $\lambda$ and $J$, as would be expected.

\subsubsection{Managed Campaign}
Importantly, the optimal sophisticated managed campaign influences the effect of competition on consumer welfare. Recall that in the data-augmented bidding model, Proposition \ref{prop:bidding_comp_statics} showed that the expected welfare of both on-platform and off-platform consumers is increasing in the number of firms $J$. However, under the best-value managed campaign, this effect is reversed for the off-platform consumers, and may be reversed for on-platform consumers as well. 

\begin{proposition}[Managed Campaign Comparative Statics]\label{prop:managed_cs}
    The best-value managed campaign price is increasing in $J$. The expected surplus of an off-platform consumer is decreasing in $J$.
\end{proposition}
\begin{proof}
 The cross-partial derivative of the vertical integration profit (multiplied by $J$) is 
 \[ \frac{\partial^2 \Pi_C}{\partial p \partial J} =  - \lambda F^J(p) \ln F(p) > 0.\]
By Topkis's theorem, the largest maximizer $p_V$ is therefore increasing in $J$. The consumer surplus comparative static follows since the off-platform consumer surplus is weakly decreasing in $p_V$. 
\end{proof}

Proposition \ref{prop:managed_cs} shows that the expected surplus of off-platform consumers decreases with the number of firms, since the off-platform prices \textit{increase} with the number of firms on the platform. 

\begin{figure}[ht]
     \centering
     \begin{subfigure}[b]{0.45\textwidth}
         \centering
         \includegraphics[width=\textwidth]{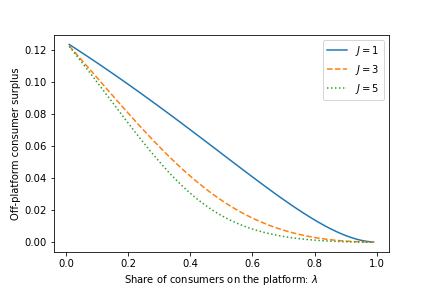}
         \caption{Expected consumer surplus off-platform}
         \label{fig:managed_off_cs}
     \end{subfigure}
     \hfill
     \begin{subfigure}[b]{0.45\textwidth}
         \centering
         \includegraphics[width=\textwidth]{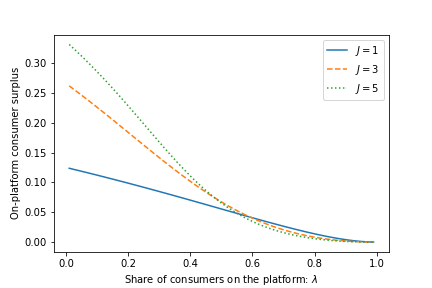}
         \caption{Expected consumer surplus on-platform}
         \label{fig:managed_on_cs}
     \end{subfigure}
     \\
     \begin{subfigure}[b]{0.45\textwidth}
         \centering
         \includegraphics[width=\textwidth]{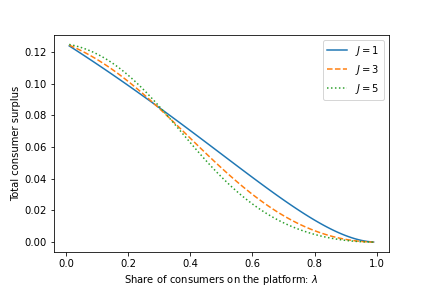}
         \caption{Total consumer surplus}
         \label{fig:managed_total_cs}
     \end{subfigure}
        \caption{Managed campaign consumer surplus and competition. Distribution of consumer values $F$ is uniform on $[0,1]$, plotted for varying $\lambda$ with $J =  1, 3, 5$.}
        \label{fig:managed_cs}
\end{figure}

On the other hand, for on-platform consumers, more firms implies that an on-platform consumer's favorite firm is likely to be more valuable (the first-order statistic is larger), which counteracts the higher off-platform prices and greater price discrimination. We present an example illustrating the tension between these two forces in Figure \ref{fig:managed_cs}. In the example, we take the consumer values to be distributed uniformly on the unit interval, and plot the expected surplus of a consumer off-platform, expected surplus of a consumer on-platform, and total consumer surplus varying the number of firms for $J = 1, 3, 5$. As indicated by Proposition \ref{prop:managed_cs}, the off-platform consumer surplus is decreasing for larger $J$, regardless of platform size. For on-platform consumers, when the platform is relatively small (lower $\lambda$), the larger off-platform market dampens on-platform price discrimination, and the surplus of on-platform consumers increases with more firms, as the on-platform consumers gain from having more firms to choose between. However, when the platform is relatively large, the increase in price discrimination with more firms results in reduced surplus for on-platform consumers despite having more firms to choose from. As such, the total consumer surplus is increasing in $J$ for small platforms (low $\lambda$) but decreasing in $J$ for large platforms (high $\lambda$).

\subsubsection{Independent Campaigns vs. Bidding}
 
We  now discuss the  implications of independent managed campaigns relative to data-augmented bidding for both posted price and welfare. We compare the independent managed campaign induced by perfect price discrimination (pricing policy $p_i(v)  = \min(v, p_i)$) to data-augmented bidding.

\begin{proposition}[Price and Welfare Comparisons]\label{prop:comparison_appendix}
Suppose the monopoly profit function $p (1 - F(p))$ is concave. 
If $F^{J-1}$ is convex, then  $ p_{I} \le p_{B}$, and total welfare and total consumer surplus are higher in the independent managed campaign than in data-augmented bidding. If $F^{J-1}$ is concave, all the inequalities are reversed. 
\end{proposition}
\begin{proof}
If the monopoly profit function is concave, then $p_I$ and $p_B$ are the unique solutions to their respective first-order conditions. Recall that $p_B$ satisfies 
\[(1-\lambda )(1-F(p_{B})-p_{B}f(p_{B}))+\lambda J\int_{p_{B}}^{\bv} F^{J-1}(v-p_{B})dF(v)=0, \]
and $p_I$ satisfies
\[(1-\lambda )(1-F(p_{I})-p_{I}f(p_{I}))+\lambda J\int_{p_{I}}^{\bv} (F^{J-1}(v)-p_{I} dF^{J-1}(v))\mathbbm{1}[p_I \in P(v;p_I)] dF(v)=0. \]
Under the assumption that $F^{J-1}$ is convex, then we have that 
\[ F^{J-1}(v - p) \ge F^{J-1}(v) - p dF^{J-1}(v). \]
since the right-hand side is a first-order expansion of $F^{J-1}$ around $v$. 
Additionally, $F^{J-1}(v - p) \ge 0$. 
Therefore, the left-hand sides of the first-order condition for $p_I$ is lower than the first-order condition for $p_B$, and hence $p_B \ge p_I$.

Note that if $F^{J-1}$ is concave, then the inequality order reverses:
\[ 0 \le F^{J-1}(v - p) \le F^{J-1}(v) - p dF^{J-1}(v), \]
since the right-hand side is a first-order expansion of $F^{J-1}$ around $v$. As $F^{J-1}(v) - p dF^{J-1}(v)$ is positive, it follows that $P(v;p_I) = \{ p_I \}$, since the maximizing function is always increasing; hence, $p_I$ solves the first-order condition 
\[(1-\lambda )(1-F(p_{I})-p_{I}f(p_{I}))+\lambda J\int_{p_{I}}^{\bv}(F^{J-1}(v)-p_{I} dF^{J-1}(v)) dF(v)=0. \]
Using the same argument as before, $p_B \le p_I$. As total welfare and total consumer surplus decrease in $p$, the welfare comparative statics follow. 
\end{proof}

Proposition \ref{prop:comparison_appendix} shows that allowing the platform to run an independent managed campaign can create a more competitive environment relative to data-augmented bidding; the threat of poaching is larger, and the competition for on-platform consumers dominates.

To interpret the condition that $F^{J-1}$ is convex, note that $F^{J-1}$ represents the cumulative distribution function of the maximum of $J-1$ values drawn from $F$. For large enough $J$,  this cumulative distribution function is  convex under relatively weak conditions. Indeed, if the density $f$ is such that $f'/f$ is bounded below, then there always exists a $J$ large enough such that $F^{J-1}$ is convex. 



We now discuss the implications of independent managed campaigns for platform revenue. Intuitively, since the joint profit of the platform and firms increases with posted prices up to $p_V$, the platform revenue ordering between bidding and the independent campaign should follow the price ranking. More precisely,

\begin{proposition}[Revenue Comparison]\label{prop:rev-comp}
    If $p_B \ge p_I$, platform revenue in a bidding model with participation fees is weakly higher than in the independent managed campaign. Otherwise, the platform earns less in the bidding model with participation fees relative to the independent managed campaign.
\end{proposition}
\begin{proof}
    Note that in both models, the firms are held to their outside options. Hence, whether the platform earns more depends exactly on the producer surplus extracted. By Theorem \ref{thm:best_value_max}, the off-platform price $p_V$ induced by the sophisticated managed campaign maximizes producer surplus. By Theorem \ref{thm:managed_prices_highest_variant}, $p_V \ge p_B, p_I$. Since producer surplus is concave in the off-platform price, $p_V$ maximizes producer surplus, and $p_V \ge p_B, p_I$, the producer surplus is larger in the bidding model iff $p_B \ge p_I$. 
\end{proof}
\begin{figure}[htbp]
    \centering
     \begin{subfigure}[b]{0.495\textwidth}
         \centering
         \includegraphics[width=\textwidth]{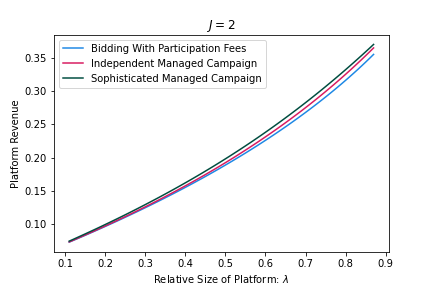}
         \caption{Platform revenue as a function of $\lambda$, $J=2$}
         \label{fig:platform_revenue_2_firms}
     \end{subfigure}
     \hfill
     \begin{subfigure}[b]{0.495\textwidth}
         \centering
         \includegraphics[width=\textwidth]{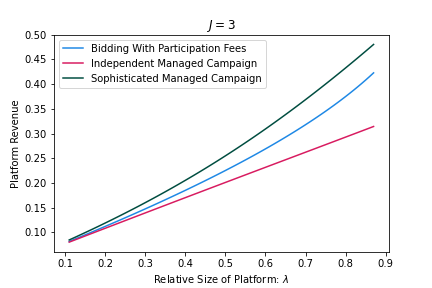}
         \caption{Platform revenue as a function of $\lambda$, $J=3$}
         \label{fig:platform_revenue_3_firms}
     \end{subfigure}
    \caption{Platform revenue with consumer values distributed as $F(v) = v^{3/4}$. With $J=2$, $F^{J-1}$ is concave, and with $J=3$, $F^{J-1}$ is convex. In the first figure for $J=2$, the platform revenue from the independent managed campaign lies in between the sophisticated managed campaign revenue and bidding with participation fees. In the second figure, the relative ordering of bidding and independent campaigns are switched.}
    \label{fig:platform_revenue}
\end{figure}
In Figure \ref{fig:platform_revenue}, we plot the revenue generated by the platform in the bidding model and the independent managed campaign as functions of $\lambda$ when consumer values are drawn from value distribution $F(v) = v^{3/4}$. Figure \ref{fig:platform_revenue_2_firms} shows the revenue when there are $J=2$ firms, and Figure \ref{fig:platform_revenue_3_firms} plots the revenue for $J=3$ firms. Figure \ref{fig:platform_revenue_2_firms} demonstrates a scenario where the independent managed campaign yields more revenue,  and Figure \ref{fig:platform_revenue_3_firms} demonstrates a case where data-augmented bidding yields more revenue. However, if we allow the platform to charge participation fees, it is clear the platform earns more revenue than in the standard bidding model without a participation fee. It is also true that in a bidding model with participation fees, the platform earns more than in an independent managed campaign. Finally, the sophisticated managed campaign results in higher platform revenue than the independent managed campaign and bidding, as would be expected by Theorem \ref{thm:best_value_max}.

\subsubsection{Cohort Privacy}

\begin{figure}[htbp]
\centering
\begin{subfigure}[b]{0.495\textwidth}
         \centering
         \includegraphics[width=\textwidth]{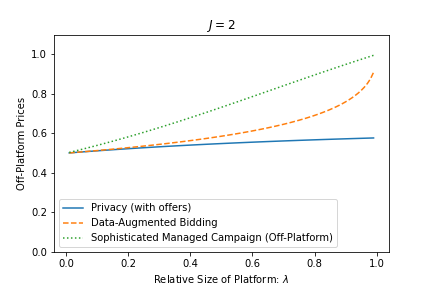}
         \caption{Off-platform prices as a function of $\lambda$, $J=2$}
         \label{fig:priv_prices_2_firms}
     \end{subfigure}
\hfill 
\begin{subfigure}[b]{0.495\textwidth}
         \centering
         \includegraphics[width=\textwidth]{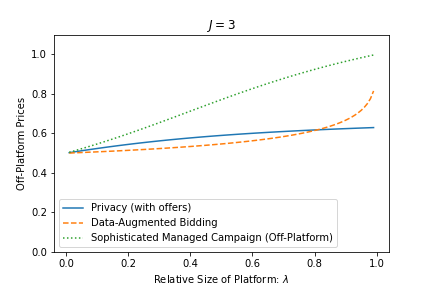}
         \caption{Off-platform prices as a function of $\lambda$, $J=3$}
         \label{fig:priv_prices_3_firms}
    \end{subfigure}
\caption{Off-platform prices under the privacy restriction, with and without
sponsored offers, compared to the benchmarks (sophisticated managed campaign
and data-augmented bidding). Prices are plotted as a function of $\protect%
\lambda$ for a uniform distribution of values and $J=2,3$. In the large $%
\protect\lambda$ limit for $J=2$, the no-offer price is larger than the
with-offer price.}
\label{fig:priv_prices}
\end{figure}
Figure \ref{fig:priv_prices} depicts the posted prices for a uniform
distribution of values, for $J = 2, 3$ firms. The plots vary the share of
on-platform consumers $\lambda$. Note that for a uniform distribution, the
monopoly price is $p_M = 0.5$. As would be expected from Proposition \ref%
{prop:partial_privacy}, the sophisticated managed campaign price is highest,
and the privacy price with offers $p_P$ is above $p_M =0.5$ but below the
managed campaign price $p_V$. However, the relative ordering of the bidding
price and the privacy price is ambiguous: for $J=3$, on smaller platforms
the bidding price $p_B$ can be lower than $p_P$. Intuitively, for
sufficiently many firms, the competitive effect of bidders on each other
profits sufficiently outweighs the incentives to raise prices. The welfare
implications of privacy are more ambiguous, as we plot in Figure \ref%
{fig:priv_welfare} for a uniform distribution of consumer values.

\begin{figure}[htbp]
\centering
\begin{subfigure}[b]{0.495\textwidth}
         \centering
         \includegraphics[width=\textwidth]{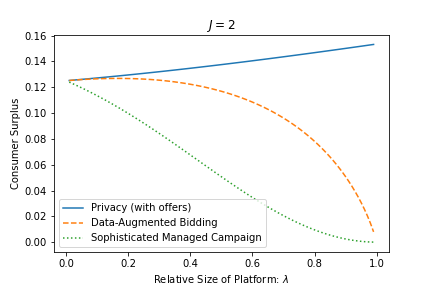}
         \caption{Consumer surplus as a function of $\lambda$, $J=2$}
         \label{fig:priv_cs_2_firms}
     \end{subfigure}
\hfill 
\begin{subfigure}[b]{0.495\textwidth}
         \centering
         \includegraphics[width=\textwidth]{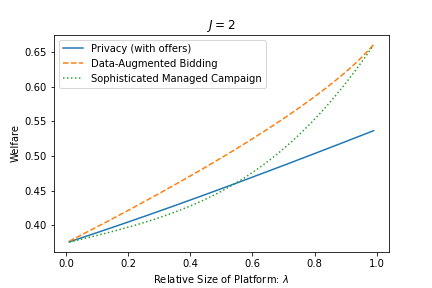}
         \caption{Total welfare as a function of $\lambda$, $J=2$}
         \label{fig:priv_welfare_2_firms}
     \end{subfigure}
\caption{Welfare implications of privacy, with a uniform distribution of
consumer values. }
\label{fig:priv_welfare}
\end{figure}

\begin{figure}[htbp]
    \centering

     \begin{subfigure}[b]{0.495\textwidth}
         \centering
         \includegraphics[width=\textwidth]{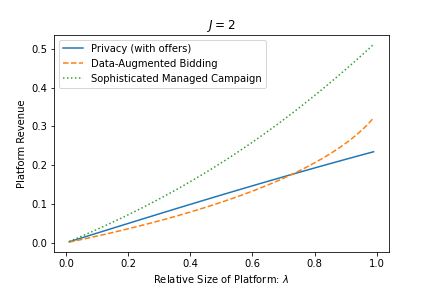}
         \caption{Platform revenue, $J=2$}
         \label{fig:priv_platform_revenue_3}
     \end{subfigure}
     \hfill
     \begin{subfigure}[b]{0.495\textwidth}
         \centering
         \includegraphics[width=\textwidth]{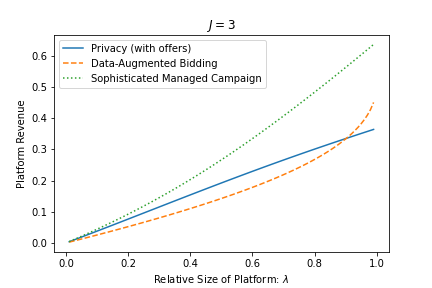}
         \caption{Platform revenue, $J = 3$}
         \label{fig:priv_platform_revenue_2}
     \end{subfigure}
     \caption{Privacy implications on platform revenue, $J=2,3$.}
     \label{fig:privacy_platform_revenue}
\end{figure}
\subsection{Robustness, RET, and Participation Fees}

\subsubsection{Bidding Auction Format}\label{sec:fpa}
In the main text, we chose a second-price auction for convenience. The main insight of Theorem \ref{thm:bidding_eq} extend to the case of the first-price auction.
 \begin{proposition}\label{prop:first_price}
     Suppose the platform ran a first-price auction instead of a second-price auction (breaking ties in favor of firms with higher value). In the unique symmetric equilibrium in undominated strategies, the firms post price $p_B$ satisfying the same condition \eqref{eqn:bidding_price} as in Theorem \ref{thm:bidding_eq}. 
 \end{proposition}
 \begin{proof}
     To show that the first-price auction induces the same posted price outcome as in Theorem \ref{thm:bidding_eq}, we argue first that Proposition \ref{prop:bidding_efficiency} still holds. Fix any vector of off-platform prices $\bar{p}$. Define $\underline{u}(v, \bar{p})$ as in \eqref{eq:outside}, and once again note that this is the outside option of a consumer on the platform. Hence, each $j$ must concede at least this much utility to the consumer, and can hence charge at most $p_j(v) = (v_j - \underline{u}(v, \bar{p}))_+$. As we argued, since $\underline{u}$ is common to all firms, the highest-valued firm is able to charge the most. In particular, suppose firm $j'$ has the second-highest value, and $j$ the highest value: $v_j > v_{j'}$. Then $j$ always wins in a first-price auction by bidding arbitrarily close to, but above $(v_{j'} - \underline{u}(v,\bar{p}))_+$. Since the platform breaks ties in favor of the firm with the highest value for the consumer, the firm $j$ bids $(v_{j'} - \underline{u}(v,\bar{p}))_+$ and wins.

However, we just showed that firm $j$ bids exactly firm $j'$'s value, and so pays exactly the same that firm $j$ would pay in the second-price auction, implementing the same allocation as the second-price auction. Since the payments are the same for the firms, the rest of the analysis in the proof of Theorem \ref{thm:bidding_eq} follows, so the equilibrium posted price outcome is the same as in \eqref{eqn:bidding_price}.
 \end{proof}

\subsubsection{Bidding with Participation Fees}
In the bidding model introduced in the main text, the platform received revenues only from the bids of the advertisers. We now ask whether tools from optimal auction design (namely, participation fees) may increase the revenue of the platform.\footnote{The importance of such tools in online ad auctions has been widely documented, e.g., by \cite{ossc23} for the case of reserve prices.} In particular, as advertisers have no prior information about the consumers, we investigate how a participation fee for the second-price auction would affect the division of surplus between the platform and advertisers.  Thus, we consider the following game:

\begin{enumerate}
    \item The platform sets a participation fee $T$.
    \item The firms choose whether to pay the participation fee and set their posted prices.  
    \item All firms observe participation decisions and posted prices. The platform runs a second-price auction for each on-platform consumer among the participating  firms. 
\end{enumerate}
The platform maximizes revenue, and we will assume a firm that is indifferent about accepting chooses to accept. As such, the platform extracts all the producer surplus, up to an outside option the firm could obtain by refusing to participate.

\begin{proposition}[Equilibrium with Participation Fees]\label{prop:bidding_participation}
    In equilibrium, all firms join and the off-platform posted prices are given by \eqref{eqn:bidding_price}. Firms bid the same as in the bidding equilibrium. Firm profits are held to their outside option \eqref{eqn:outside}.
    The fee charged by the platform holds firms to this outside option.
\end{proposition}
\begin{proof}
    By Proposition \ref{prop:bidding_efficiency}, the firm willing to pay the most for any consumer regardless of off-platform prices is the firm which the consumer has the highest value for; hence, it is not revenue optimal for the platform to exclude any firm from participating. Consider the subgame after all firms have paid the participation fee. Subgame perfection and Theorem \ref{thm:bidding_eq} imply that the pricing condition for off-platform prices is given by \eqref{eqn:bidding_price}, and firms bid their true value $\max(v_j, p_B)$. It is then straightforward to see that the maximum participation fee must hold the firm's profit to what they could get from being excluded. Consider the profit firm $i$ makes after refusing participation, with all other firms joining the platform. By Proposition \ref{prop:bidding_efficiency}, the firm can never get a sale from a consumer whose favorite firm is $i \neq j$; further, the highest bidder for a consumer whose favorite firm is $i$ is exactly firm $j$ such that $j = \arg\max_{j \neq i} v_j$. Note that $v_j$ can undercut the posted price firm $i$ sets by pricing down to (possibly) zero; so firm $i$ can only retain on-platform sales from the consumers such that $v_i - p \ge v_j - 0$. That is, the on-platform sales are 
    \[ \int_p^{\bv} p F^{J-1}(v - p) dF(v). \]
    Combining, we find the exclusion profit is given by \eqref{eqn:outside}, and the result follows. 
\end{proof}

Intuitively, the pricing and bidding behavior follow as in Theorem \ref{thm:bidding_eq} due to subgame perfection. The fee charged is as large as possible to make firms indifferent between accepting and rejecting, and thus holds firms to their outside option profit \eqref{eqn:outside}. Qualitatively, the participation fee does not change the posted prices, but redistributes surplus from firms towards the platform. 
\subsubsection{Managed Campaign: Variable Fees}

Alternatively to fixed fees, the platform could charge a commission. That is, instead of taking a fixed fee for each sale, the platform could instead take a fraction $\alpha$ of the revenue of each sale. However, we show that this does result in a difference in the pricing incentives relative to Theorem \ref{thm:managed_eq_best_value}.

\begin{proposition}[Commissions]\label{prop:commissions}
Suppose the platform charged a commission of $\alpha \in (0, 1)$ of the revenue from all sales. Then under best-value pricing, the subgame equilibrium off-platform price set by the firms under commissions satisfies 
\begin{equation} \label{eq:managed_price_commissions}
    0 = \frac{1 - \lambda}{ J}(1 - F(p_{\alpha}) - p f(p_\alpha)) +  \lambda(1-\alpha)\left( \int_{p_{\alpha} } F^{J-1}(v) \ dF(v)  \right).
\end{equation}
Further, $p_\alpha < p_V$, where $p_V$ is the equilibrium price from Theorem \ref{thm:managed_eq_best_value}.
\end{proposition}
\begin{proof}
    Note that with a commission $\alpha$ charged on all sales, the  firm's best-response pricing profit condition becomes 
\[ \Pi_\alpha(p, p') = \frac{1 - \lambda }{J}p(1 - F(p)) +  \lambda (1-\alpha)\int_{\uv}^{\bv}\min(v,p) F^{J-1}(v) \ dF(v). \]
for $p < p'$, since the firm only captures the remaining $(1-\alpha)$ of the on-platform sale value. Since the cross-derivative of $\Pi_\alpha$ in $p, \alpha$ is negative, by Topkis's theorem, it follows that the price solving this first-order condition must be lower than the best-value price $p_V$. 
\end{proof}

Proposition \ref{prop:commissions} shows that by charging commissions, the platform does induce a distortion relative to $p_V$; however, this distortion is better for consumers, since the distortion lowers the off-platform prices relative to the managed campaign in Theorem \ref{thm:managed_eq_best_value}. Intuitively, with commissions, the firm does not have as strong an incentive to raise its prices, since the firm's marginal benefit is diminished (multiplied by $1-\alpha$). 

However, this insight is central to showing that a capped commission structure (i.e., the platform charges a commission, but never charges more than a capped amount $t$ per sale) can yield the Theorem \ref{thm:managed_eq_best_value} outcome; that is, the platform could charge commissions on low-value sales, but implement a cap to induce the firms' to price higher off the platform.

\begin{proposition}[Commissions with Caps]\label{prop:commissions_cap}
    Suppose the platform charged commission of $\alpha \in (0,1)$ of the revenue from all sales, but capped its maximum commission fee at $t$. Then if $t/\alpha \le p_\alpha$, then the equilibrium posted price is the same $p_V$ from \eqref{eq:managed_price_best_value} as in Theorem \ref{thm:managed_eq_best_value}.
 \end{proposition}
 \begin{proof}
     With commission $\alpha$ capped at some $t_j$, the firm's on-platform profit changes. If $p \ge t_j / \alpha$, the profit is
\[ \int_{\uv}^{t_j/\alpha} (1-\alpha) v F^{J-1}(v) dF(v) + \int_{t_j/\alpha}^p (v - t)F^{J-1}(v)dF(v) + \int_p^{\bv} (p-t)F^{J-1}(v)dF(v). \]
If $p < t_j/\alpha$,
\[ \int_{\uv}^{p} (1-\alpha) v F^{J-1}(v) dF(v) + \int_{p}^{\bv} (1-\alpha)p F^{J-1}(v)dF(v), \]
since the cap never binds in this case. In the first case, note that the derivative of the profit term is exactly 
\[ \int_p^{\bv} F^{J-1}(v) dF(v). \]
which is the same term as in best-value pricing. To argue that the firm never wants to set a price in the second case, it suffices to see that the profit function is strictly increasing in this range. Thus, if $t_j / \alpha < p_\alpha$ (where $p_\alpha$ is the commissions price we characterized in Proposition \ref{prop:commissions}), the profit function of the firm must be strictly increasing on this range. Hence, the firm must set a price at least $t_j / \alpha$, and so the firm's best response profit is maximized by $p_V$.
 \end{proof}

 Intuitively, the cap removes the distortion on the marginal benefit from pricing higher induced by commissions; thus, the implementation of a cap can actually harm consumers, as the cap induces higher firm pricing.

\end{document}